\newif\ifarxiv%
\arxivtrue%

\documentclass[runningheads,envcountsect,envcountsame,orivec]{llncs}

\usepackage{silence}
\usepackage{hyperref}
\usepackage{microtype}
\usepackage{xspace}
\usepackage{enumerate}
\usepackage{paralist}
\usepackage{thm-restate}

\usepackage{amsmath}
\usepackage{amssymb}
\usepackage{mathpartir}
\usepackage{stmaryrd}

\newcommand{\KA}{\ensuremath{\mathsf{KA}}\xspace}
\newcommand{\F}{\ensuremath{\mathsf{F}_1}\xspace}
\newcommand{\SKA}{\ensuremath{\mathsf{SKA}}\xspace}

\newcommand{\NSF}{\ensuremath{\mathsf{NSF}}\xspace}
\newcommand{\SL}{\ensuremath{\mathsf{SL}}\xspace}
\newcommand{\SF}{\ensuremath{\mathsf{SF}_1}\xspace}

\newcommand{\equivka}{\equiv_{\scriptscriptstyle\KA}}
\newcommand{\equivf}{\equiv_{\scriptscriptstyle\F}}
\newcommand{\equivska}{\equiv_{\scriptscriptstyle\SKA}}
\newcommand{\equivsf}{\equiv_{\scriptscriptstyle\SF}}

\newcommand{\equivsem}{\equiv_{\scriptscriptstyle\SL}}
\newcommand{\notequivska}{\not\equiv_{\scriptscriptstyle\SKA}}

\newcommand{\terms}{\ensuremath{\mathcal{T}}}
\newcommand{\termsreg}{\terms_{\KA}}
\newcommand{\termsregf}{\terms_{\F}}
\newcommand{\termssf}{\terms_{\SF}}
\newcommand{\termssynreg}{\terms_{\SKA}}
\newcommand{\termsskanf}{\terms_{\NSF}}

\newcommand{\sem}[1]{\llbracket#1\rrbracket}
\newcommand{\semka}[1]{\sem{#1}_{\scriptscriptstyle\KA}}
\newcommand{\semf}[1]{\sem{#1}_{\scriptscriptstyle\F}}
\newcommand{\semska}[1]{\sem{#1}_{\scriptscriptstyle\SKA}}
\newcommand{\semsf}[1]{\sem{#1}_{\scriptscriptstyle\SF}}

\newcommand{\semsem}[1]{\sem{#1}_{\scriptscriptstyle\SL}}

\newcommand{\powerset}[1]{\mathcal{P}(#1)} 
\newcommand{\continuation}{o}
\newcommand{\reach}{\rho}

\newcommand{\pipe}{\;\;|\;\;}

\newcommand{\twopartdef}[4]{\left\{\begin{array}{ll}#1 & #2 \\#3 & #4\end{array}\right.}

\newcommand{\bact}{\Sigma}
\newcommand{\cross}{\terms_{\SL}}
\newcommand{\nfcross}{\overline{\cross}}

\newcommand{\pipr}[1]{(#1)}
\newcommand{\piprs}[1]{#1}
\newcommand{\piproduct}[1]{\pipr{#1}^{\Pi}}
\newcommand{\piproducts}[1]{\piprs{#1}^{\Pi}}

\newcommand{\synchronouslanguages}{\mathcal{L}_{\bact}}
\newcommand{\modela}{s}
\newcommand{\countermodel}{\mathcal{L}_{\modela}}

\ifarxiv%
\usepackage[top=43mm,bottom=43mm,left=46mm,right=46mm]{geometry}
\fi%

\begin{document}
\title{%
   Completeness and Incompleteness of Synchronous Kleene~Algebra%
    \texorpdfstring{\thanks{%
       This work was partially supported by ERC Starting Grant ProFoundNet (679127), a Leverhulme Prize (PLP--2016--129) and a Marie Curie Fellowship (795119).
         The first author conducted part of this work at Centrum Wiskunde \& Informatica, Amsterdam.
    }}{}
}

\author{%
    Jana Wagemaker \inst{1}
    \and
    Marcello Bonsangue \inst{2}
    \and
    Tobias Kapp\'{e}\inst{1}
    \and\texorpdfstring{\\}{}
    Jurriaan Rot\inst{1,3}
    \and
    Alexandra Silva\inst{1}
}

\authorrunning{J. Wagemaker \and M. Bonsangue \and T. Kapp\'{e} \and J. Rot \and A. Silva}
\institute{
    University College London, London, United Kingdom
    \and
    Leiden University, Leiden, The Netherlands
    \and
    Radboud University, Nijmegen, The Netherlands
}

\maketitle
\begin{abstract}
 \emph{Synchronous Kleene algebra} (\emph{SKA}), an extension of Kleene algebra (KA), was proposed by Prisacariu as a tool for reasoning about programs that may execute synchronously, i.e., in lock-step.
We provide a countermodel witnessing that the axioms of SKA are incomplete w.r.t.\ its language semantics, by exploiting a lack of interaction between the \emph{synchronous product} operator and the Kleene star.
We then propose an alternative set of axioms for SKA, based on Salomaa's axiomatisation of regular languages, and show that these provide a sound and complete characterisation w.r.t.\ the original language semantics.
\end{abstract}

\section{Introduction}
\emph{Kleene algebra} (\emph{KA}) is applied in various contexts, such as relational algebra and automata theory. An important use of KA is as a logic of programs. This is because the axioms of KA correspond well to properties expected of sequential program composition, and hence they provide a logic for reasoning about control flow of sequential programs presented as Kleene algebra expressions. Regular languages then provide a canonical semantics for programs expressed in Kleene algebra, due to a tight connection between regular languages and the axioms of KA:\ an equation is provable using the Kleene algebra axioms if and only if the corresponding regular languages coincide~\cite{boffa-1990,krob-1991,kozen-1994}.

In~\cite{prisacariu}, Prisacariu proposes an extension of Kleene algebra, called \emph{synchronous Kleene algebra} (\emph{SKA}). The aim was to introduce an algebra useful for studying not only sequential programs but also \emph{synchronous} concurrent programs. Here, synchrony is understood as in Milner's SCCS~\cite{milner}, i.e., each program executes a single action instantaneously at each discrete time step. Hence, the synchrony paradigm assumes that basic actions execute in one unit of time and that at each time step, all components capable of acting will do so.
This model permits a \emph{synchronous product} operator, which yields a program that, at each time step, executes some combination of the actions put forth by the operand programs.

This new operator is governed by various expected axioms such as associativity and commutativity. Another axiom describes the interaction between the synchronous product and the sequential product, capturing the intended lock-step behaviour. Crucially, the axioms do not entail certain equations that relate the Kleene star (used to describe loops) and the synchronous product.

The contributions of this paper are twofold. First, we show that the lack of connection between the Kleene star and the synchronous product is problematic. In particular, we exploit this fact to devise a countermodel that violates a semantically valid equation, thus showing that the SKA axioms are incomplete w.r.t.\ the language semantics. This invalidates the completeness result in~\cite{prisacariu}.

The second and main contribution of this paper is a sound and complete characterisation of the equational theory of SKA in terms of a generalisation of regular languages. The key difference with~\cite{prisacariu} is the shift from \emph{least} fixpoint axioms in the style of Kozen~\cite{kozen-1994} to a \emph{unique} fixpoint axiom in the style of Salomaa~\cite{salomaa}. In the completeness proof, we give a reduction to the completeness result of Salomaa via a normal form for SKA expressions. As a by-product, we get a proof of the correctness of the partial derivatives for SKA provided in~\cite{broda}.

This paper is organised as follows. In \autoref{preliminaries} we discuss the necessary preliminaries. In \autoref{preliminaries:ska} we discuss SKA as presented in~\cite{prisacariu}. Next, in \autoref{incompleteness}, we demonstrate why SKA is incomplete, and in \autoref{newaxiomatisation} go on to provide a new set of axioms, which we call \SF. The latter section also includes basic results about the partial derivatives for SKA from~\cite{broda}. 
In \autoref{fundamental} we provide an algebraic characterisation of \SF-terms; this characterisation is used in~\autoref{completeness}, where we prove completeness of \SF w.r.t.\ to its language model. In \autoref{section:related-work} we consider related work and conclude by discussing directions for future work in \autoref{section:future-work}. For the sake of readability, some of the proofs appear in the appendix.

\section{Preliminaries}\label{preliminaries}
Throughout this paper, we write $2$ for the two-element set $\{ 0, 1 \}$.

\paragraph{Languages}
Throughout the paper we fix a finite alphabet $\Sigma$. A \emph{word} formed over $\Sigma$ is a finite sequence of symbols from $\Sigma$. The \emph{empty word} is denoted by $\varepsilon$. We write $\Sigma^*$ for the set of all words over $\Sigma$. \emph{Concatenation} of words $u, v \in \Sigma^*$ is denoted by $uv \in \Sigma^*$. A \emph{language} is a set of words. For $K,L\subseteq \Sigma^*$, we
define
\begin{align*}
& K\cdot L= \{uv : u\in K, v\in L\}
&& K+L=K\cup L
&& K^*= \bigcup\nolimits_{n\in\mathbb{N}} K^n,
\end{align*}
where $K^0=\{\varepsilon\}$ and $K^{n+1}=K\cdot K^n$.

\paragraph{Kleene Algebra}
We define a \emph{Kleene algebra}~\cite{kozen-1994} as a tuple $(A,+,\cdot,^*,0,1)$ where $A$ is a set, $^*$ is a unary operator, $+$ and $\cdot$ are binary operators and $0$ and $1$ are constants. Moreover, for all $e,f,g\in A$ the following axioms are satisfied:
\begin{align*}
& e + (f + g) = (e + f) + g
&& e + f = f + e
&& e + 0 = e \qquad e + e = e
\\
&e \cdot 1 = e = 1 \cdot e
&& e \cdot 0 = 0 = 0 \cdot e
&&e \cdot (f \cdot g) = (e \cdot f) \cdot g
\\
&e^* = 1 + e \cdot e^* = 1 + e^* \cdot e
&& (e + f) \cdot g = e \cdot g + f \cdot g
&& e \cdot (f + g) = e \cdot f + e \cdot g
\end{align*}
Additionally, we write $e \leq f$ as a shorthand for $e + f = f$, and require that the \emph{least fixpoint axioms}~\cite{kozen-1994} hold, which stipulate that for $e, f, g \in A$ we have
\begin{mathpar}
e + f \cdot g \leq g \implies f^* \cdot e \leq g
\and
e + f \cdot g \leq f \implies e \cdot g^* \leq f
\end{mathpar}

\noindent
The set of \emph{regular expressions}, denoted $\termsreg$, is described by the grammar:
\[
    \termsreg\ni e,f ::= 0 \pipe 1 \pipe a \in \Sigma \pipe e + f \pipe e \cdot f \pipe e^*
\]
Regular expressions can be interpreted in terms of languages.
This is done by defining $\semka{-}: \termsreg\rightarrow \powerset{\Sigma^*}$ inductively, as follows.
\begin{align*}
&\semka{0} = \emptyset
&& \semka{a} = \{ a \}
&& \semka{e \cdot f} = \semka{e} \cdot \semka{f} \\
&\semka{1} = \{\varepsilon\}
&&\semka{e + f} = \semka{e} + \semka{f}
&&\semka{e^*} = \semka{e}^*
\end{align*}
A language $L$ is called \emph{regular} if and only if $L=\semka{e}$ for some $e\in\termsreg$.%

We write $\equivka$ for the smallest congruence on $\termsreg$ induced by the Kleene algebra axioms --- e.g., for all $e \in \termsreg$, we have $1 + e \cdot e^* \equivka e^*$.
Intuitively, $e \equivka f$ means that the regular expressions $e$ and $f$ can be proved equivalent according to the axioms of Kleene algebra.
A pivotal result in the study of Kleene algebras tells us that $\semka{-}$ characterises $\equivka$, in the following sense:
\begin{theorem}[Soundness and Completeness of KA~\cite{kozen-1994}]%
\label{theorem:sandcka}
For all $e,f\in\termsreg$, we have that $e\equivka f$ if and only if $\semka{e}=\semka{f}
$.
\end{theorem}

\begin{remark}%
    \label{remark:extended-soundness-ka}
    The above can be generalised, as follows.
    Let $\mathcal{K} = (A, +, \cdot, ^*, 0, 1)$ be a KA, and let $\sigma: \Sigma \to A$.
    Then for all $e, f \in \termsreg$ such that $e \equivka f$, interpreting $e$ and $f$ according to $\sigma$ in $\mathcal{K}$ yields the same result.
    For instance, since ${(a^*)}^* \equivka a^*$, we know that for \emph{any} element $e$ of \emph{any} KA $\mathcal{K}$, we have that ${(e^*)}^* = e$.
\end{remark}

\paragraph{Linear Systems}\label{linearsystems}
Let $Q$ be a finite set. A \emph{$Q$-vector} is a function $x: Q\rightarrow\termsreg$. A \emph{$Q$-matrix} is a function $M: Q\times Q \rightarrow\termsreg$. Let $x$ and $y$ be $Q$-vectors. Addition is defined pointwise, setting $(x+y)(q)=x(q)+y(q)$. Multiplication by a $Q$-matrix $M$ is given by
\[
(M\cdot x)(q)=\sum_{e\in Q}M(q,e)\cdot x(e)
\]
When $x(q) \equivka y(q)$ for all $q \in Q$, we write $x \equivka y$.

\begin{definition}\label{definition:solution-linear}
A \emph{$Q$-linear system} is a pair $(M, x)$ with $M$ a $Q$-matrix and $x$ a $Q$-vector. A solution to $(M, x)$ in KA is a $Q$-vector $y$ such that $M\cdot y + x\equivka y$.
\end{definition}

\paragraph{Non-deterministic finite automata}
A \emph{non-deterministic automaton (NDA)} over an alphabet $\Sigma$ is a triple $(X,o,d)$
where $o \colon X \rightarrow 2$ is called the \emph{termination function} and $d \colon X \times \Sigma \rightarrow X$ called the \emph{continuation function}. If $X$ is finite, $(X,o,d)$ is referred to as a \emph{non-deterministic finite automaton (NFA)}.

The semantics of an NDA $(X,o,d)$ can be characterised recursively as the unique map $\ell: X \to \powerset{\Sigma^*}$ such that
\begin{equation}\label{uniquesemantics}
\ell(x) = \{ \varepsilon : o(x) = 1 \} \cup \bigcup_{x' \in d(x, a)} \{ a \} \cdot \ell(x')
\end{equation}
This coincides with the standard definition of language acceptance.

\section{Synchronous Kleene Algebra}\label{preliminaries:ska}
Synchronous Kleene algebra extends Kleene algebra with an additional operator denoted $\times$, which we refer to as the synchronous product~\cite{prisacariu}.

\begin{definition}[Synchronous Kleene Algebra]\label{definition:ska}
A \emph{synchronous KA (SKA)} is a tuple $(A, S, +, \cdot, ^*,\times, 0,1)$ such that $(A, +, \cdot, ^*, 0,1)$ is a Kleene algebra and $\times$ is a binary operator on $A$, with $S \subseteq A$ closed under $\times$ and $(S,\times)$ a semilattice.
Furthermore, the following hold for all $e,f,g\in A$ and $\alpha,\beta\in S$:
  \begin{align*}
  & e \times (f + g) = e \times f + e \times g
  && \quad e \times (f \times g) = (e \times f) \times g
  && \quad e \times 0 = 0 \\
  & (\alpha\cdot e)\times (\beta\cdot f) = (\alpha\times \beta)\cdot (e\times f)
  && \quad e \times f = f\times e
  && \quad e \times 1 =  e
\end{align*}
\end{definition}

\noindent Note that $0$ and $1$ need not be elements of $S$.
The \emph{semilattice terms}, denoted $\cross$, are given by the following grammar.
\[
    \cross \ni e,f :: = a\in\bact \pipe e\times f
\]
The \emph{synchronous regular terms}, denoted $\termssynreg$, are given by the grammar:
\[
    \termssynreg\ni e, f ::= 0 \pipe 1 \pipe a \in \cross \pipe e + f \pipe e \cdot f \pipe e\times f \pipe e^*
\]
Thus we have $\cross\subseteq\termssynreg$. We then define $\equivska$ as the smallest congruence on $\termssynreg$ satisfying the axioms of \SKA\@. Here, $\cross$ plays the role of the semilattice; for instance, for $a \in \cross$ we have that $a \times a \equivska a$.

\begin{remark}
In~\cite{prisacariu}, $\times$ is declared to be idempotent on the \emph{generators} of the semilattice, whereas in our definition it holds for semilattice elements in general. This does not change anything, as the axiom $a \times a = a$ for generators together with commutativity and associativity results in idempotence on the semilattice. We present SKA as in \autoref{definition:ska} to prevent a meta-definition of a third sort (namely the semilattice generated by $\bact$) present in the signature of the algebra. We have also left out the second distributivity and unit axioms that follow immediately from the ones presented and commutativity.
\end{remark}

\subsection{A Language Model for SKA}
Similar to Kleene algebra, there is a language model for SKA~\cite{prisacariu}.

Words over $\mathcal{P}(\bact)\setminus\{\emptyset\}=\mathcal{P}_n(\bact)$ are called \emph{synchronous strings}, and sets of synchronous strings are called \emph{synchronous languages}. The standard language operations (sum, concatenation, Kleene closure) are also defined on synchronous languages. The synchronous product of synchronous languages $K, L$ is given by:
\[
K\times L = \{u\times v : u\in K, v\in L\}
\]
where we define $\times$ inductively for $u, v \in {(\mathcal{P}_n(\bact))}^*$ and $x, y \in \mathcal{P}_n(\bact)$, as follows:
\begin{mathpar}
u\times \varepsilon = u = \varepsilon \times u
\and
\text{and}
\and
(x\cdot u)\times( y\cdot v) = (x\cup y)\cdot (u\times v)
\end{mathpar}

To define the language semantics for all elements in $\termssynreg$, we first give an interpretation of elements in $\cross$ in terms of non-empty finite subsets of $\bact$.
\begin{definition}
For $a\in\bact$ and $e,f\in\cross$, define $\semsem{-}:\cross\rightarrow \mathcal{P}_n(\bact)$ by
\begin{mathpar}
\semsem{a}=\{a\}
\and
\semsem{e\times f} = \semsem{e}\cup \semsem{f}
\end{mathpar}
\end{definition}

Denote the smallest congruence on $\cross$ with respect to idempotence, associativity and commutativity of $\times$ with $\equivsem$. It is not hard to show that $\semsem{-}$ characterises $\equivsem$, in the following sense.

\begin{lemma}[Soundness and Completeness of SL]\label{semilatticestuff}
For all $e,f\in\cross$, we have $\semsem{e}=\semsem{f}$ if and only if $e\equivsem f$.
\end{lemma}

The semantics of synchronous regular terms is given in terms of a mapping to synchronous languages:
$\semska{-}:\termssynreg \rightarrow \powerset{{(\mathcal{P}_n(\bact))}^*}$. We have:
 \[\small
 \begin{array}{@{}l@{\hspace{.5cm}}l@{\hspace{.5cm}}l@{}}
   \semska{0} = \emptyset\quad\semska{1} = \{\varepsilon\}
     & \semska{a} = \{ \semsem{a} \} \quad \forall a\in\cross
     & \semska{e^*} = \semska{e}^*\\[1ex]
\semska{e \cdot f} = \semska{e} \cdot \semska{f}
     & \semska{e + f} = \semska{e} + \semska{f}
     &  \semska{e \times f} = \semska{e} \times \semska{f}
 \end{array}
\]
A synchronous language $L$ is called \emph{regular} when $L=\semska{e}$ for some $e\in\termssynreg$.

Let $S=\{\{x\} : x\in \mathcal{P}_n(\bact)\}$, that is to say, $S$ is the set of synchronous languages consisting of a single word, whose single letter is in turn a subset of $\bact$. Furthermore, let $\synchronouslanguages$ denote the set of synchronous languages over $\bact$. It is straightforward to prove that $\synchronouslanguages$ together with $S$ is closed under the SKA operations and satisfies the SKA axioms~\cite{prisacariu}; more precisely, we have:

\begin{restatable}{lemma}{synlanska}%
\label{lemma:synlan_ska}
The structure $(\synchronouslanguages,S,+, \cdot ,^*,\times, \emptyset,\{\varepsilon\})$ is an SKA, that is, synchronous languages over $\bact$ form an SKA\@.
\end{restatable}

As a consequence of \autoref{lemma:synlan_ska}, we obtain soundness of the SKA axioms with respect to the language model based on synchronous regular languages:

\begin{restatable}[Soundness of SKA]{lemma}{soundness}%
\label{lemma:soundness}
For all $e, f \in \termssynreg$, we have that $e \equivska f$ implies $\semska{e} = \semska{f}$.
\end{restatable}

\begin{remark}%
    \label{remark:extended-soundness-ska}
    The above generalises almost analogously to \autoref{remark:extended-soundness-ka}.
    Let $\mathcal{M}$ be an SKA with semilattice $S$, and let $\sigma: \Sigma \to S$ be a function.
    Then for all $e, f \in \termssynreg$ such that $e \equivska f$, if we interpret $e$ in $\mathcal{M}$ according to $\sigma$, then we should get the same result as when we interpret $f$ in $\mathcal{M}$ according to $\sigma$.

    In other words, when $e \equivska f$ holds, it follows that $e = f$ is a valid equation in \emph{every} SKA, provided that the symbols from $\Sigma$ are interpreted as elements of the semilattice.
    It is not hard to show that this claim does not hold when symbols from $\Sigma$ can be interpreted as elements of the carrier at large.
\end{remark}

\section{Incompleteness of SKA}\label{incompleteness}
We now prove incompleteness of the SKA\@ axioms as presented in~\cite{prisacariu}. Fix alphabet $\mathcal{A}=\{a\}$. First, note that the language model of SKA has the following property.

\begin{restatable}{lemma}{astar}\label{astar}
For $\alpha\in\cross$, we have $\semska{\alpha^* \times \alpha^*}=\semska{\alpha^*}$.
\end{restatable}

If $\equivska$ were complete w.r.t. $\semska{-}$, then the above implies that $a^*\times a^* \equivska a^*$ holds. In this section, we present a countermodel where all the axioms of SKA are true, but $\alpha^*\times \alpha^* = \alpha^*$ does not hold for
any $\alpha \in S$. This shows that $a^* \times a^* \not\equivska a^*$; consequently, $\equivska$ cannot be complete w.r.t. $\semska{-}$.

\subsection*{Countermodel for SKA}\label{section:countermodel}
We define our countermodel as follows. For the semilattice, let $S=\{\{\{ \modela \}\}\}$, the set containing the synchronous language $\{\{ \modela \}\}$. We denote the set of all synchronous languages over alphabet $\{ \modela \}$ with $\countermodel$; the carrier of our model is formed by $\countermodel \cup \{ \dagger \}$, where $\dagger$ is a symbol not found in $\countermodel$. The symbol $\dagger$ exists only in the model, and not in the algebraic theory. It remains to define the SKA operators on this carrier, which we do as follows.

\begin{definition}\label{definitionoperators}
An element of $\countermodel\cup\{\dagger\}$ is said to be infinite when it is an infinite language. For $K,L \in \countermodel\cup\{\dagger\}$, define the SKA operators as follows:
\[
    \begin{array}{rll}
        K + L
        &=
        \left\{
            \begin{array}{l}
                \dagger \\[1mm]
                K \cup L
            \end{array}
        \right.
        &
        \begin{array}{l}
            K = \dagger \vee L = \dagger \\[1mm]
            \text{otherwise}
        \end{array} \\[5mm]
        K \cdot L
        &=
        \left\{
            \begin{array}{l}
                \emptyset \\[1mm]
                \dagger \\[1mm]
                \{ u \cdot v : u \in K, v \in L \}
            \end{array}
        \right.
        &
        \begin{array}{l}
            K = \emptyset \vee L = \emptyset \\[1mm]
            K = \dagger \vee L = \dagger \\[1mm]
            \text{otherwise}
        \end{array} \\[8mm]
        K \times L 
        &=
        \left\{
            \begin{array}{l}
                \emptyset \\[1mm]
                \dagger \\[1mm]
                \{ u \times v : u \in K, v \in L \}
            \end{array}
        \right.
        &
        \begin{array}{l}
            K = \emptyset \vee L = \emptyset \\[1mm]
            K = \dagger \vee L = \dagger \vee K, L \text{ infinite} \\[1mm]
            \text{otherwise}
        \end{array} \\[8mm]
        K^*
        &=
        \left\{
            \begin{array}{l}
                \dagger \\[1mm]
                \bigcup_{n \in \mathbb{N}} K^n
            \end{array}
        \right.
        &
        \begin{array}{l}
            K = \dagger \\[1mm]
            \text{otherwise}
        \end{array}
    \end{array}
\]
where $u\times v$ for $u\in K$ and $v\in L$ and $K^n$ is as defined in \autoref{preliminaries:ska}.
Here, the cases are given in order of priority --- e.g., if $K=\emptyset$ and $L=\dagger$, then $K \cdot L = \emptyset$.
\end{definition}

The intuition behind this model is that SKA has no axioms that relate to the synchronous execution of starred expressions, such as in $\alpha^* \times \alpha^*$, nor can such a relation be derived from the axioms, meaning that a model has some leeway in defining the outcome in such cases. Since the language of a starred expression is generally infinite, we choose $\times$ such that it diverges to the extra element $\dagger$ when given infinite languages as input; for the rest of the operators, the behaviour on $\dagger$ is chosen to comply with the axioms.

First, we verify that our operators satisfy the SKA axioms.

\begin{restatable}{lemma}{countermodelis}\label{lemma:countermodelisska}
$\mathcal{M}=(\countermodel\cup\{\dagger\},\{\{\{ \modela \}\}\},+,\cdot,^*,\times,\emptyset,\{\varepsilon\})$ with the operators as defined in \autoref{definitionoperators} forms an SKA\@.
\end{restatable}
\begin{proof}
For the sake of brevity, we validate one of the least fixpoint axioms and the synchrony axiom; the other axioms are treated in the appendix.

Let $K,L,J\in\countermodel\cup\{\dagger\}$. We verify that $K+L\cdot J \leq J \implies L^*\cdot K\leq J$. Assume that $K+L\cdot J\leq J$. If $J=\dagger$, then the result follows by definition of $\leq$ and our choice of $+$. Otherwise, if $J\in \countermodel $, we distinguish two cases. If $L=\dagger$, then $J$ must be $\emptyset$ (otherwise $J = \dagger$); hence $K = \emptyset$, and the claim holds. Lastly, if $L\in\countermodel$, then $K\in\countermodel$. In this case, all of the operands are languages, and thus the proof goes through as it does for KA\@.

For the synchrony axiom, we need only check
\[
    (A\cdot K)\times (A\cdot L)=(A\times A)\cdot (K\times L)
\]
for $A=\{\{ \modela \}\}$ as that is the only element in $S$. Let $K,L\in \countermodel\cup\{\dagger\}$. If either $K$ or $L$ is $\emptyset$, both sides of the equation reduce to $\emptyset$. Otherwise, if $K$ or $L$ is $\dagger$, then both sides of the equation reduce to $\dagger$.
If $K$ and $L$ are both infinite then $A\cdot K$ and $A\cdot L$ are infinite and the claim follows. In all the remaining cases where $K$ and $L$ are elements of $\countermodel$ and at most one of them is infinite, the proof goes through as it does for synchronous regular languages (\autoref{lemma:soundness}).
\qed
\end{proof}

This leads us to the following theorem:
\begin{theorem}
The axioms of SKA presented in \autoref{definition:ska} are incomplete. That is, there exist
$e,f\in\termssynreg$ such that $\semska{e}=\semska{f}$ but $e\notequivska f$.
\end{theorem}
\begin{proof}
Take $a\in\mathcal{A}$. We know from \autoref{astar} that $\semska{a^* \times a^*}=\semska{a^*}$.
Now suppose $a^*\times a^*\equivska a^*$. As our countermodel is an SKA that means in particular that ${\{\{ \modela \}\}}^* \times {\{\{ \modela \}\}}^* = {\{\{ \modela \}\}}^*$ should hold (c.f. \autoref{remark:extended-soundness-ska}).
However, in this model we can calculate that ${\{\{ \modela \}\}}^* \times {\{\{ \modela \}\}}^* = \dagger\neq {\{\{ \modela \}\}}^*$.
Hence, we have a contradiction. Thus $a^*\times a^*\notequivska a^*$, rendering SKA incomplete.\qed
\end{proof}

\section{A new axiomatisation}\label{newaxiomatisation}
We now create an alternative algebraic formalism, which we call \SF, and prove that its axioms are sound and complete w.r.t the model of synchronous regular languages. Whereas the definition of SKA relies on Kleene algebras (with \emph{least fixpoint axioms}) as presented by Kozen~\cite{kozen-1994}, the definition of $\SF$ builds on $\F$-algebras (with a \emph{unique fixpoint axiom}) as presented by Salomaa~\cite{salomaa}. The axioms of Salomaa are strictly stronger than Kozen's~\cite{struthfoster}, and we will see that the unique fixpoint axiom allows us to derive a connection between the synchronous product and the Kleene star, even though this connection is not represented in an axiom directly (see \autoref{heiligvoorbeeld}).

\begin{definition}
An \emph{$\F$-algebra}~\cite{salomaa} is a tuple $(A,+,\cdot,^*,0,1,H)$ where $A$ is a set, $^*$ is a unary operator, $+$ and $\cdot$ are binary operators and $0$ and $1$ are constants, and such that for all $e,f,g\in A$ the following axioms are satisfied:
\begin{align*}
& e + (f + g) = (e + f) + g
&& e + f = f + e
&& e + 0 = e \qquad e + e = e
\\
&e \cdot 1 = e = 1 \cdot e
&& e \cdot 0 = 0 = 0 \cdot e
&&e \cdot (f \cdot g) = (e \cdot f) \cdot g
\\
&e^* = 1 + e \cdot e^* = 1 + e^* \cdot e
&& (e + f) \cdot g = e \cdot g + f \cdot g
&& e \cdot (f + g) = e \cdot f + e \cdot g
\end{align*}
Additionally, the \emph{loop tightening} and \emph{unique fixpoint axiom} hold:
\begin{align*}
& {(e + 1)}^* = e^*
&& H(f) = 0 \wedge e + f \cdot g = g \implies f^* \cdot e = g
\end{align*}
Lastly, we have the following axioms for $H$:
\begin{align*}
& H(0) = 0
&& H(e + f) = H(e) + H(f)
&& H(e^*)={(H(e))}^* \\
& H(1) = 1
&& H(e \cdot f) = H(e) \cdot H(f)
\end{align*}
\end{definition}
In~\cite{salomaa}, an $e \in A$ with $H(e) = 1$ is said to have the \emph{empty word property}, which will be reflected in the semantic interpretation of $H(e)$ stated below.

The set of \emph{$\F$-expressions}, denoted $\termsregf$, is described by:
\[
    \termsregf\ni e,f ::= 0 \pipe 1 \pipe a \in \Sigma \pipe e + f \pipe e \cdot f \pipe e^* \pipe H(e)
\]
We can interpret $\F$-expressions in terms of languages through $\semf{-}: \termsregf \to \powerset{\Sigma^*}$, defined analogously to $\semka{-}$, where furthermore for $e \in \termsregf$ we have
\[
    \semf{H(e)} = \semf{e} \cap \{ \varepsilon \}
\]

We write $\equivf$ for the smallest congruence on $\termsregf$ induced by the $\F$-axioms.
Additionally, we require that for $a \in \Sigma$, we have $H(a) \equivf 0$.
A characterisation similar to \autoref{theorem:sandcka} can then be established as follows\footnote{Unlike~\cite{salomaa}, we include $H$ in the syntax; one can prove that for any $e \in \termsregf$ it holds that $H(e) \equiv 0$ or $H(e) \equiv 1$, and hence any occurence of $H$ can be removed from $e$. This is what allows us to apply the completeness result from op.\ cit.\ here.}:
\begin{theorem}[Soundness and Completeness of \F~\cite{salomaa}]\label{salomaacomplete}
For all $e, f \in \termsregf$, we have that $e \equivf f$ if and only if $\semf{e} = \semf{f}$.
\end{theorem}

\begin{remark}%
    \label{remark:extended-soundness-f1}
    Kozen~\cite{kozen-1994} noted that the above does not generalise along the same lines as in \autoref{remark:extended-soundness-ka}.
    In particular, the axiom $H(a) \equivska 0$ is not stable under substitution; for instance, if we interpret $H(a)$ according to the valuation $a \mapsto \{ \epsilon \}$ in the $\F$-algebra of languages, then we obtain $\{ \epsilon \}$, whereas $0$ is interpreted as $\emptyset$.
\end{remark}

\begin{definition}
A \emph{synchronous \F-algebra} (\emph{$\SF$-algebra} for short) is a tuple $(A,S,+,\cdot,^*,0,1,H)$, such that $(A, +, \cdot,^*, 0, 1, H)$ is an $\F$-algebra and $\times$ is a binary operator on $A$, with $S \subseteq A$ closed under $\times$ and $(S, \times)$ a semilattice. Furthermore, the following hold for all $e, f, g \in A$ and $\alpha,\beta\in S$:
\begin{align*}
  & e \times (f + g) = e \times f + e \times g
  && \quad e \times (f \times g) = (e \times f) \times g
  && \quad e \times 0 = 0 \\
  & (\alpha\cdot e)\times (\beta\cdot f) = (\alpha\times \beta)\cdot (e\times f)
  && \quad e \times f = f\times e
  && \quad e \times 1 =  e
\end{align*}
Moreover, $H$ is compatible with $\times$ as well, i.e., for $e, f \in A$ we have that $H(e \times f) = H(e) \times H(f)$.
Lastly, for $\alpha \in S$ we require that $H(\alpha) = 0$.
\end{definition}

\begin{remark}
The countermodel from \autoref{section:countermodel} cannot be extended to a model of $\SF$.
To see this, note that we have $H(\{\{\modela\}\}) = 0$ and $\emptyset + \{\{\modela\}\} \cdot \dagger = \dagger$, but ${\{\{\modela\}\}}^* \cdot \emptyset \neq \dagger$ --- contradicting the unique fixpoint axiom.
\end{remark}

\noindent The set of \emph{$\SF$-expressions} over $\bact$, denoted $\termssf$, is described by:
\[
    \termssf\ni e,f ::= 0 \pipe 1 \pipe a \in \cross \pipe e + f \pipe e \cdot f \pipe e\times f \pipe e^* \pipe H(e)
\]
We interpret $\termssf$ in terms of languages via $\semsf{-}: \termssf \rightarrow \synchronouslanguages$, defined analogously to $\semska{-}$, where furthermore for $e\in\termssf$ we have
\[
\semsf{H(e)} = \semsf{e}\cap \{\varepsilon\}
\]
Note that when $e \in \termssynreg$, then $e \in \termssf$ and $\semska{e} = \semsf{e}$.

Define $\equivsf$ as the smallest congruence on $\termssf$ induced by the axioms of \SF, where $\cross$ fulfills the role of the semilattice --- e.g., if $a \in \cross$, then $a \times a \equivsf a$. This axiomatisation is sound with respect to the language model.\footnote{
    Note that for the synchronous language model we know the least fixpoint axioms are sound as well (\autoref{lemma:soundness}). However, there might be other $\SF\text{-models}$ where the least fixpoint axioms are not valid.
}

\begin{restatable}{lemma}{soundnessska}\label{soundness:unique}
Let $e,f\in\termssf$. If $e\equivsf f$ then $\semsf{e}=\semsf{f}$.
\end{restatable}

\begin{remark}
    The caveat from \autoref{remark:extended-soundness-f1} can be transposed to this setting.
    However, the condition that for $\alpha \in S$ we have that $H(\alpha) = 0$ allows one to strengthen the above along the same lines as \autoref{remark:extended-soundness-ska}, that is, if $e \equivsf f$, then interpreting $e$ and $f$ in some SKA according to some valuation of $\Sigma$ in terms of semilattice elements will produce the same outcome.
\end{remark}

\begin{remark}\label{heiligvoorbeeld}
To demonstrate the use of the new axioms, we give an algebraic proof of $\alpha^*\times \alpha^*\equivsf \alpha^*$ for $\alpha\in \cross$:
\begin{align*}
  \alpha^*\times \alpha^* & \equivsf (1+\alpha\cdot \alpha^*)\times (1+\alpha\cdot \alpha^*) \equivsf 1+\alpha\cdot \alpha^* + (\alpha\cdot \alpha^*)\times (\alpha\cdot \alpha^*) \\
  &\equivsf 1+ \alpha\cdot \alpha^* + (\alpha\times \alpha)\cdot (\alpha^*\times \alpha^*) \equivsf \alpha^* + \alpha\cdot (\alpha^*\times \alpha^*)
\end{align*}
Since $H(\alpha) = 0$, we can apply the unique fixpoint axiom to find $\alpha^*\cdot \alpha^*\equivsf \alpha^*\times \alpha^*$. In $\SF$, it is not hard to show that $\alpha^*\cdot \alpha^* \equivf \alpha^*$; hence, we find $\alpha^*\times \alpha^*\equivsf \alpha^*$.
\end{remark}

\begin{remark}\label{isditeenremark?}
Adding $\alpha^*\times \alpha^* = \alpha^*$ for $\alpha\in \cross$ as an axiom to the old axiomatisation of SKA would not have been sufficient; one can easily find another semantical truth that does not hold in our countermodel, such as $\semska{{(\alpha\cdot\beta)}^*\times{(\alpha\cdot\beta)}^*}=\semska{{(\alpha\cdot\beta)}^*}$. Adding $e^*\times e^* = e^*$ as an axiom is also not an option, as this is not sound; for instance, take $e=a+b$ for $a,b\in\bact$. In order to keep the axiomatisation finitary, a unique fixpoint axiom provided an outcome.
\end{remark}

\subsection{Partial Derivatives}
In this section we develop the theory of SKA and set up the necessary machinery for \autoref{fundamental} and the completeness proof in \autoref{completeness}. We start by presenting partial derivatives, which provide a termination and continuation map on $\termssf$. These derivatives allow us to turn the set of synchronous regular terms into a non-deterministic automaton structure, such that the language accepted by $e\in\termssf$ as a state in this automaton is the same as the semantics of $e$. Furthermore, partial derivatives turn out to provide a way to algebraically characterise a term by means of acceptance and reachable terms, which is useful in the completeness proof of \SF.\@

The termination and continuation map for $\SF$-expressions presented below are a trivial extension of the ones from~\cite{broda}. Intuitively, the termination map is $1$ if an expression can immediately terminate, and $0$ otherwise; the continuation map of a term w.r.t.\ $A$ gives us the set of terms reachable with an $A$-step.

\begin{definition}[Termination map]
  For $a\in\bact$, we define $\continuation: \termssf \to 2$ inductively, as follows:
  \begin{align*}
  \continuation(0) &= 0  & \continuation(e^*) &= 1   & \continuation(e + f)     &=
  \max(\continuation(e),\continuation(f)) & \continuation (e \times f) &= \min(\continuation(e), \continuation(f)) \\
  \continuation(1) &= 1  & \continuation(a) &= 0  & \continuation(e \cdot f) &= \min(\continuation(e), \continuation(f)) & \continuation (H(e)) & = \continuation (e)
  \end{align*}
\end{definition}

\begin{definition}[Continuation map]
  For $a\in\bact$, we inductively define \quad\quad $\delta: \termssf\times\mathcal{P}_n(\bact) \to \powerset{\termssf}$ as follows:
\begin{align*}
\delta(0, A) &= \delta(1,A) = \emptyset
&  \delta(e \times f, A) &= \Delta(e,f,A)\cup \Delta(f,e,A)\\
\delta(H(e), A) &= \emptyset
   && \cup \{ e' \times f' : e' \in \delta(e, B_1), \\
\delta(a, A) &= \{ 1 : A=\{a\} \}
&& f'\in\delta(f,B_2), B_1\cup B_2=A \}  \\
\delta(e^*, A) &= \{ e' \cdot e^* : e' \in \delta(e, A) \}
& \delta(e \cdot f, A)&= \{ e' \cdot f : e' \in \delta(e, A) \} \\
 \delta(e + f, A) &= \delta(e, A) \cup \delta(f, A)
 &&  \cup \Delta(f,e,A)
\end{align*}
where $\Delta(e, f, A)$ is defined to be $\delta(e, A)$ when $\continuation(f) = 1$, and $\emptyset$ otherwise.
\end{definition}

\begin{definition}[Syntactic Automaton]\label{def:synaut}
We call the NDA $(\termssf,\continuation,\delta)$ the \emph{syntactic automaton} of $\SF$-expressions.
\end{definition}

In \autoref{fundamental} we give a proof of correctness of partial derivatives: for $e\in\termssf$ the semantics of $e$ is equivalent to the language accepted by $e$ as a state in the syntactic automaton. An analogous property holds for (partial) derivatives in Kleene algebras~\cite{brz,antimirov}, which makes derivatives a powerful tool for reasoning about language models and deciding equivalences of terms~\cite{bonchi-pous-2013}.

In the next two sections, we want to use terms reachable from $e$, that is to say, terms that are a result of repeatedly applying the continuation map on $e$. To this end, we define the following function:
\begin{definition}\label{reach}
For $e,f \in \termssf$ and $a\in\bact$, we inductively define the \emph{reach} function $\reach: \termssf \to \powerset{\termssf}$ as follows:
\begin{align*}
&\reach(e + f) = \reach(e) \cup \reach(f)  && \reach(0) = \emptyset  \\
& \reach(e \cdot f) = \{ e' \cdot f : e' \in \reach(e) \} \cup \reach(f)  && \reach(1) = \{ 1 \}  \\
 &\reach(e^*) = \{ 1 \} \cup \{ e' \cdot e^* : e' \in \reach(e) \}
 && \reach(a) = \{ 1, a \} \\
&\reach(e\times f) = \{e'\times f' : e'\in\reach(e),f'\in\reach(f)\} \cup\reach(e)\cup\reach(f)
&& \reach(H(e))=\{1\}
\end{align*}
\end{definition}
\noindent Using a straightforward inductive argument, one can prove that for all $e \in \termssf$, $\reach(e)$ is finite. Note that $e$ is not always a member of $\reach(e)$.
To see that $\rho(e)$ indeed contains all terms reachable from $e$, we record the following.
\begin{restatable}{lemma}{reachisreach}\label{lemma:reachisreach}
For all $e\in\termssf$ and $A\in\mathcal{P}_n(\bact)$, we have $\delta(e,A)\subseteq\reach(e)$. Also, if $e'\in\reach (e)$, then $\delta(e',A)\subseteq\reach(e)$.
\end{restatable}

\subsection{Normal form}\label{section:normalform}
In this section we develop a \emph{normal form} for expressions in $\cross$, which we will use in the completeness proof for \SF. As $\semsem{-}$ is a surjective function it has at least one right inverse. Let us pick one and denote it by $\piproduct{-}$. We thus have $\piproduct{-}:\mathcal{P}_n(\bact)\rightarrow\cross$ such that $\semsem{-}\circ\piproduct{-}$ is the identity on $\mathcal{P}_n(\Sigma)$. 

The normal form for expressions in $\cross$ is defined as follows:
\begin{definition}[Normal form]\label{normalform}
For all $e\in\cross$ the \emph{normal form} of $e$, denoted as $\overline{e}$, is defined as $\piproduct{\semsem{e}}$. Let $\nfcross=\{\overline{e} : e\in\cross \}$.
\end{definition}

Intuitively, an expression in normal form is standardised with respect to idempotence, associativity and commutativity. For instance, for a term $(a\times a)\times (c\times b)$ with $a,b,c\in\bact$, the chosen normal form, dictated by the chosen right inverse, could be $(a\times b)\times c$, and all terms provably equivalent to $(a\times a)\times (c\times b)$ will have this same normal form. Using \autoref{semilatticestuff}, we can formalise this in the following two results:
\begin{lemma}\label{factsnormalform}
For all $e\in\cross$, we have that $e$ is provably equivalent to its normal form: $e\equivsem\overline{e}$. Moreover, if two expressions $e,f\in\cross$ are provably equivalent, they have the same normal form: if $e\equivsem f$, then $\overline{e}=\overline{f}$.
\end{lemma}
\begin{proof}
As $\piproduct{-}$ is a right inverse of $\semsem{-}$, we can derive the following:
\[
\semsem{\overline{e}}=\semsem{\piproduct{\semsem{e}}}=\semsem{e}
\]
From completeness we get $e\equivsem\overline{e}$. For the second part of the statement we obtain via soundness that $\semsem{e}=\semsem{f}$ and subsequently that $\overline{e}=\overline{f}$.\qed
\end{proof}

Normalising normalised terms does not change anything.
\begin{restatable}{lemma}{nfisnf}\label{lemma:nfisnf}
For all $e\in\nfcross$ we have that $\overline{e}=e$.
\end{restatable}

We extend ${(-)}^{\Pi}$ from synchronous strings of length one to words and synchronous languages in the obvious way. For a synchronous string $aw$ with $a\in\mathcal{P}_n(\bact)$ and $w\in {(\mathcal{P}_n(\bact))}^*$, and synchronous language $L\in \synchronouslanguages$ we define:
\begin{mathpar}
\varepsilon^{\Pi}=\varepsilon
\and
{(aw)}^{\Pi}=a^{\Pi}\cdot (w^{\Pi})
\and
L^{\Pi}=\{w^{\Pi} : w\in L\}
\end{mathpar}
We abuse notation and assume the type of ${(-)}^{\Pi}$ is clear from the argument.

Since ${(-)}^\Pi$ is a homomorphism of languages, we have the following.

\begin{restatable}{lemma}{pionlanguages}\label{lemma:pi-on-languages}
For synchronous languages $L$ and $K$, all of the following hold:
\begin{inparaenum}[(i)]
    \item\label{lemma:plus}
    ${(L\cup K)}^{\Pi}=L^{\Pi}\cup K^{\Pi}$,
    \item\label{lemma:sequential}
    ${(L\cdot K)}^{\Pi}=L^{\Pi}\cdot K^{\Pi}$, and
    \item\label{lemma:star}
    ${(L^*)}^{\Pi}= {(L^{\Pi})}^*$.
\end{inparaenum}
\end{restatable}

\section{A Fundamental Theorem for \SF}\label{fundamental}
In this section we shall algebraically capture an expression in terms of its partial derivatives. This characterisation of an \SF-term will be useful later on in proving completeness but also provides us with a straightforward method to prove correctness of the partial derivatives. Following~\cite{rutten:2003,thesisalexandra}, we call this characterisation a \emph{fundamental theorem} for \SF.\@ Before we state and prove the fundamental theorem, we prove an intermediary lemma:

\begin{lemma}\label{lemma:sommetjes}
For all $e,f\in\termssf$, we have
\[
    \sum_{e'\in\delta(e,A)} (A^{\Pi}\cdot e')\times  \sum_{e'\in\delta(f,A)} (A^{\Pi}\cdot e')\equivsf
    \sum_{\substack{e'\in\delta(e,A) \\ e'' \in \delta(f, B)}} {(A\cup B)}^{\Pi}\cdot (e'\times e'')
\]
\end{lemma}
\begin{proof}
First note the following derivation for $A,B\in\mathcal{P}_n(\bact)$, using \autoref{factsnormalform}, the fact that all axioms of $\equivsem$ are included in $\equivsf$, and that $\piproduct{-}$ is a right inverse of $\semsem{-}$:
\begin{align*}
  \piproducts{A}\times \piproducts{B} &\equivsf \overline{\piproducts{A}\times \piproducts{B}} = \piproduct{\semsem{\piproducts{A}\times \piproducts{B}}} \\
  &=\piproduct{\semsem{\piproducts{A}}\cup\semsem{ \piproducts{B}}}=\piproduct{A\cup B}
\end{align*}
Using distributivity, the synchrony axiom and the equation above, we can derive:
\begin{align*}
  &\sum_{e'\in\delta(e,A)} (A^{\Pi}\cdot e')\times  \sum_{e'\in\delta(f,A)} (A^{\Pi}\cdot e')
  & \equivsf \sum_{\substack{e'\in\delta(e,A) \\ e'' \in \delta(f, B)}} (A^{\Pi}\cdot e')\times (B^{\Pi}\cdot e'')\\
  & \equivsf \sum_{\substack{e'\in\delta(e,A) \\ e'' \in \delta(f, B)}} (A^{\Pi}\times B^{\Pi})\cdot (e'\times e'')
  & \equivsf \sum_{\substack{e'\in\delta(e,A) \\ e'' \in \delta(f, B)}} \piproduct{A\cup B}\cdot (e'\times e'')
\end{align*}
The synchrony axiom can be applied because $\piproducts{A},\piproducts{B}\in\cross$.\qed
\end{proof}

\begin{restatable}[Fundamental Theorem]{theorem}{fundamentaltheorem}\label{theorem:fundamental}
  For all $e\in\termssf$, we have
  \[e\equivsf \continuation(e)+\sum_{e'\in\delta(e,A)} A^{\Pi}\cdot e'.\]
\end{restatable}
\begin{proof}
This proof is mostly analogous to the proof of the fundamental theorem for regular expressions, such as the one that can be found in~\cite{thesisalexandra}.

We proceed by induction on $e$. In the base, we have three cases to consider: $e\in\{0,1\}$ or $e=a$ for $a\in\bact$. For $e\in \{0,1\}$, the result follows immediately.
For $e=a$, the only non-empty derivative is $\delta(a,\{a\})$ and the result follows:
\[
    \continuation(a)+\sum_{e'\in\delta(a,A)} A^{\Pi}\cdot e'\equivsf \continuation(a) + \overline{a}\cdot 1 \equivsf \overline{a} \equivsf a
\]
In the inductive step, we treat only the case for synchronous composition; the others can be found in the appendix.
If $e=e_0\times e_1$, derive as follows:
\begin{align*}
& e_0\times e_1 \\
    &\equivsf \big(\continuation(e_0)+\sum_{e'\in\delta(e_0,A)} A^{\Pi}\cdot e'\big) \times \big(\continuation(e_1)+\sum_{e'\in\delta(e_1,A)} A^{\Pi}\cdot e'\big)
        \tag{Ind. hyp.} \\
    &\equivsf \continuation(e_0)\times \continuation(e_1) +\sum_{e'\in\delta(e_0,A)} (A^{\Pi}\cdot e')\times \continuation(e_1)+ \continuation(e_0)\times \sum_{e'\in\delta(e_1,A)} A^{\Pi}\cdot e'\\
    & \quad \quad +\sum_{e'\in\delta(e_0,A)} (A^{\Pi}\cdot e')\times  \sum_{e'\in\delta(e_1,A)} (A^{\Pi}\cdot e')
    \tag{Distributivity}   \\
    & \equivsf \continuation(e_0\times e_1) +\sum_{e'\in\delta(e_0,A)} (A^{\Pi}\cdot e')\times \continuation(e_1)
     + \continuation(e_0)\times \sum_{e'\in\delta(e_1,A)} A^{\Pi}\cdot e' \\
    & \quad \quad +\sum_{\substack{e'\in\delta(e_0,A) \\ e'' \in \delta(e_1, B)}} {(A\cup B)}^{\Pi}\cdot (e'\times e'')
    \tag{Def. $\continuation$, \autoref{lemma:sommetjes}} \\
    & \equivsf \continuation(e_0\times e_1) +\sum_{e'\in\Delta(e_0,e_1,A)\hspace{-1cm}} A^{\Pi}\cdot e' + \sum_{e'\in\Delta(e_1,e_0,A)\hspace{-1cm}} A^{\Pi}\cdot e'
    + \sum_{\substack{ e'\in\{e_0'\times e_1':e_0'\in\delta(e_0,A), \\ e_1'\in\delta(e_1,B),C=A\cup B\}}\hspace{-1cm}} C^{\Pi}\cdot e'\\
    & \equivsf \continuation(e_0\times e_1) +\sum_{e'\in\delta(e_0\times e_1,A)} A^{\Pi}\cdot e' \tag*{(Def. $\delta$) \qed}
\end{align*}
\end{proof}

\subsection*{Correctness of partial derivatives for \SF}
We now relate the partial derivatives for \SF to their semantics.
Let $\ell:\termssf\rightarrow \synchronouslanguages$ be the semantics of the syntactic automaton $(\termssf,\continuation,\delta)$ (\autoref{def:synaut}), uniquely defined by \autoref{uniquesemantics}:

\begin{equation}\label{semanticsofautomaton}
\ell(e) = \{ \varepsilon : \continuation(e) = 1 \} \cup \bigcup_{e' \in \delta(e, A)} \{A\}\cdot \ell(e')
\end{equation}

To prove correctness of derivatives for \SF, we prove that the language semantics of the syntactic automaton and the \SF-expression coincide:
\begin{theorem}[Soundness of derivatives]\label{soundnessderivatives}
For all $e\in\termssf$ we have:
\[ \ell(e)=\semsf{e}\]
\end{theorem}
\begin{proof}
  The claim follows almost immediately from the fundamental theorem. From \autoref{lemma:soundness} and \autoref{theorem:fundamental}, we obtain
  \[
  \semsf{e}=\{\varepsilon:\continuation(e)=1\}\cup\bigcup_{e'\in\delta(e,A)} \{A\}\cdot \semsf{e'}
  \]
  Note that $\semsf{A^{\Pi}}=\{\semsem{\piproducts{A}}\}=\{A\}$ by definition of the \SF semantics of a term in $\cross$ and the fact that $\piproduct{-}$ is a right inverse. Because $\ell$ is the only function satisfying \autoref{semanticsofautomaton}, we obtain the desired equality between $\semsf{e}$ and the language $\ell(e)$ accepted by $e$ as a state of the automaton $(\termssf,\continuation,\delta)$.
\qed
\end{proof}

\section{Completeness of \SF}\label{completeness}
In this section we prove completeness of the \SF-axioms with respect to the synchronous language model: we prove that for $e,f\in\termssf$, if $\semsf{e}=\semsf{f}$, then $e\equivsf f$.
We first prove completeness of \SF for a subset of $\SF$-expressions, relying on the completeness result of \F\@ (\autoref{lemma:completenessnormalform}). Then we demonstrate that for every \SF-expression we can find an equivalent \SF-expression in this specific subset (\autoref{thm:normalform}). This subset is formed as follows.

\begin{definition}\label{def:nf}
The set of $\SF$-expressions in \emph{normal form}, $\termsskanf$, is described by the grammar
\[
    \termsskanf\ni e, f ::= 0 \pipe 1 \pipe  a \in \nfcross \pipe e + f \pipe e \cdot f \pipe e^*
\]
where $\nfcross$ is as defined in \autoref{normalform}.
\end{definition}

From this description it is immediate that an \SF-term $e\in\termsskanf$ is formed from terms of $\nfcross$ connected via the regular \F-algebra operators. Hence, $\F$-expressions formed over the alphabet $\nfcross$ are the same set of terms as $\termsskanf$. We shall use this observation to prove completeness for $\termsskanf$ with respect to the language model by leveraging completeness of \F\@.

We use the function $\piproduct{-}$ to give a translation between the \SF semantics of a term in $\termsskanf$ and the \F semantics of that same term:

\begin{restatable}{lemma}{semantictranslation}\label{soundnessf}
For all $e\in\termsskanf$, we have ${(\semsf{e})}^{\Pi}=\semf{e}$.
\end{restatable}
\begin{proof}
We proceed by induction on the construction of $e$. In the base, there are three cases to consider. If $e=0$, then $\semsf{e}=\emptyset=\semf{e}$, and we are done. If $e=1$, then ${(\semsf{e})}^{\Pi}={(\{\varepsilon\})}^{\Pi}=\{\varepsilon\}=\semf{1}$, and the claim follows.
If $e=a$ for $a\in\nfcross$, we use \autoref{lemma:nfisnf} to obtain $\overline{a}=a$.
As $a\in\nfcross\subseteq\cross$, we know that ${(\semsf{a})}^{\Pi}={(\{\semsem{a}\})}^{\Pi}=\{\piproduct{\semsem{a}}\}=\{\overline{a}\}=\{a\}=\semf{a}$, and the claim follows.

For the inductive step, first consider $e=H(e_0)$. ${(\semsf{H(e_0)})}^{\Pi}=\{\varepsilon\}$ if $\varepsilon\in\semsf{e_0}$ and $\emptyset$ otherwise. We also have $\semf{H(e_0)}=\{\varepsilon\}$ if $\varepsilon\in\semf{e_0}$ and $\emptyset$ otherwise.
The induction hypothesis states that ${(\semsf{e_0})}^{\Pi}=\semf{e_0}$, from which we obtain that $\varepsilon\in \semsf{e_0}\Leftrightarrow\varepsilon\in\semf{e_0}$. Hence we can conclude that ${(\semsf{H(e_0)})}^{\Pi}=\semf{H(e_0)}$. All other inductive cases follow immediately from \autoref{lemma:pi-on-languages}. The details can be found in the appendix.\qed
\end{proof}

We are now ready to prove completeness of \SF for terms in normal form.
\begin{lemma}\label{lemma:completenessnormalform}
Let $e,f\in\termsskanf$. If $\semsf{e}=\semsf{f}$, then $e\equivsf f$.
\end{lemma}
\begin{proof}
By the premise, we have that ${(\semsf{e})}^{\Pi}={(\semsf{f})}^{\Pi}$. From \autoref{soundnessf} we
get ${(\semsf{e})}^{\Pi}=\semf{e}$ and ${(\semsf{f})}^{\Pi}=\semf{f}$, which results in  $\semf{e}=\semf{f}$.
From \autoref{salomaacomplete} we know that this entails that $e\equivf f$. As \SF contains all the axioms of \F, we may then conclude that $e\equivsf f$ and the claim follows.\qed
\end{proof}

In order to prove completeness with respect to the language model for all $e\in\termssf$, we prove that for every $e\in\termssf$ there exists a term $\hat{e}\in\termsskanf$ in normal form such that $e\equivsf\hat{e}$. To see this is indeed enough to establish completeness of \SF, imagine we have such a procedure to transform $e$ into $\hat{e}$ in normal form.
We can then conclude that $\semsf{e}=\semsf{f}$ implies $\semsf{\hat{e}}=\semsf{\hat{f}}$, which by \autoref{lemma:completenessnormalform} implies $\hat{e}\equivsf\hat{f}$, and consequently that $e\equivsf f$.

To obtain $\hat{e}$, we will make use of the ``unfolding'' of an \SF-expression $e$ in terms of partial derivatives, given by the fundamental theorem, which will give rise to a linear system. We will then show that this linear system has a unique solution that has the properties we require from $\hat{e}$. Since $e$ is also a solution to this linear system, we can conclude that they are provably equivalent.

Let us start with the following property of linear systems over \SF. A $Q$-vector is a function $x: Q\rightarrow\termssf$ and a $Q$-matrix is a function $M: Q\times Q \rightarrow\termssf$. We call a matrix $M$ \emph{guarded} if $H(M(i,j))=0$ for all $i,j\in Q$. 
We say a vector $p$ and matrix $M$ are in normal form if $p(i)\in\termsskanf$ for all $i\in Q$ and  $M(i,j)\in\termsskanf$ for all $i,j\in Q$. The following lemma is a variation of~\cite[Lemma 2]{salomaa} and the proof is a direct adaptation of the proof found in~\cite[Lemma 3.12]{cka}.

\begin{restatable}{lemma}{leastsolution}\label{lemma:leastsolution}
    Let $(M, p)$ be a $Q$-linear system such that $M$ and $p$ are guarded. We can construct $Q$-vector $x$ that is the unique (up to $\SF$-equivalence) solution to $(M, p)$ in \SF. Moreover, if $M$ and $p$ are in normal form, then so is $x$. 
\end{restatable}

We now define the linear system associated to an \SF-expression $e$. This linear system makes use of the partial derivatives for \SF, and essentially represents an NFA that acceps the language described by $e$.

\begin{definition}
Let $e\in\termssf$, and choose $Q_e =\reach(e)\cup \{e\}$, where $\reach$ is the reach function from \autoref{reach}. Define the $Q_e$-vector $x_e$ and the $Q_e$-matrix $M_e$ by
\begin{mathpar}
x_e(e')=\continuation(e')
\and
M_e(e',e'')= \sum_{e''\in\delta(e',A)}A^{\Pi}
\end{mathpar}
\end{definition}

We can now use \autoref{lemma:leastsolution} to obtain the desired normal form $\hat{e}$:

\begin{theorem}\label{thm:normalform}
For all $e\in\termssf$, there exists an $\hat{e}\in\termsskanf$ such that $\hat{e}\equivsf e$.
\end{theorem}
\begin{proof}
It is clear from their definition that $x_e$ and $M_e$ are both in normal form and that $M_e$ is guarded. From \autoref{lemma:leastsolution} we then get that there exists a unique solution $s_e$ to $(M_e, x_e)$, and $s_e$ is a $Q_e$-vector in normal form.
Now consider the $Q_e$-vector $y$ such that $y(q)=q$ for all $q\in Q_e$. Using \autoref{lemma:reachisreach} and \autoref{theorem:fundamental}, we can derive the following:
\begin{align*}
x_e(q) + M_e\cdot y(q)&\equivsf x_e(q) + \sum_{q'\in Q_e}M_e(q,q')\cdot y(q') \\
&\equivsf \continuation(q)+\sum_{q'\in Q_e}\sum_{q'\in\delta(q,A)}A^{\Pi}\cdot q' \\
&\equivsf \continuation(q)+\sum_{q'\in\delta(q,A)}A^{\Pi}\cdot q' \equivsf q = y(q)
\end{align*}
This demonstrates that $y$ is also a solution to $(M_e, x_e)$. As we know from \autoref{lemma:leastsolution} that $s_e$ is unique, we get that $y\equivsf s_e$. This means that $e=y(e)\equivsf s_e(e)$. As $s_e$ is in normal form we get that $s_e(e)\in\termsskanf$. Thus, if we take $s_e(e)=\hat{e}$, then we have obtained the desired result.\qed
\end{proof}

Combining \autoref{thm:normalform} and \autoref{lemma:completenessnormalform} gives the main result of this section:
\begin{theorem}[Soundness and Completeness]\label{maintheorem}
For all $e,f\in\termssf$, we have
\[e\equivsf f \Leftrightarrow \semsf{e}=\semsf{f}
\]
\end{theorem}

As a corollary of \autoref{soundnessderivatives} and \autoref{maintheorem} we know that \SF is decidable by deciding language equivalence in the syntactic automaton.

\section{Related Work}\label{section:related-work}
Synchonous cooperation among processes has been extensively studied
in the context of process calculi such as ASP~\cite{BK1984} and
SCCS~\cite{milner}. SKA bears a strong resemblance to SCCS, with
the most notable differences being the equivalence axiomatised
(bisimulation vs.\ language equivalence), and the use of Kleene star
(unbounded finite recursion) instead of fixpoint (possibly infinite
recursion). Contrary to ASP, but similar to SCCS, SKA cannot express incompatibility of action synchronisation.

In the context of Kleene algebra based frameworks for concurrent reasoning, a synchronous product is just one possible interpretation of concurrency. An interleaving-based approach with a concurrent operator (a parallel operator denoted with $\parallel$) is explored in Concurrent Kleene Algebra~\cite{cka,struthenzo,hoare-2009,hoare-2016,cka}.

We have proved that $\equivsf$ is sound and complete with respect to the synchronous language model by making use of the completeness of \F~\cite{salomaa}. The strategy of transforming an expression $e$ to an equivalent expression $\hat{e}$ with a particular property is often used in literature~\cite{cka,kat,struthenzo,kao}. In particular, we adopted the use of linear systems as a representation of automata, which was first done by Conway~\cite{completenessmetconway} and Backhouse~\cite{backhouse-1975}.
The machinery that we used to solve linear systems in \F is based on Salomaa~\cite{salomaa} and can also be found in~\cite{cka} and~\cite{kozen-2001}. The idea of the syntactic automaton originally comes from Brzozowski, who did this for regular expressions~\cite{brz}. He worked with derivatives which turn a Kleene algebra expression into a deterministic automaton. We worked with partial derivatives instead, resulting in a non-deterministic finite automaton for each \SF-expression. Partial derivatives were first proposed by Antimirov~\cite{antimirov}.

Other related work is that of Hayes et al.~\cite{hayes1}. They explore an algebra of synchronous atomic steps that interprets the synchrony model SKA is based on (Milner's SCCS calculus). However, their algebra is not based on a Kleene algebra --- they use concurrent refinement algebra~\cite{cra} instead.  Later, Hayes et al.\ presented an abstract algebra for reasoning about concurrent programs with an abstract synchronisation operator~\cite{hayes2}, of which their earlier algebra of atomic steps is an instance. A key difference seems to be that Hayes et al.\ use different units for synchronous and sequential composition. It would be interesting to compare expressive powers of the two algebras more extensively.

A decision procedure for equivalence between SKA terms is given by Broda et al.~\cite{broda}. They defined partial derivatives for SKA that we also used in the proof of completeness, and used those to construct an NFA that accepts the semantics of a given SKA expression. Deciding language equivalence of two automata then leads to a decision procedure for semantic equivalence of SKA expressions.

\section{Conclusions and Further Work}\label{section:future-work}
We have presented a complete axiomatisation with respect to the model of synchronous regular languages. We have first proved incompleteness of SKA via a countermodel, exploiting the fact that SKA did not have any axioms relating the synchronous product to the Kleene star. We then provided a set of axioms based on the \F-axioms from Salomaa~\cite{salomaa} and the axioms governing the synchronous product familiar from SKA\@. This was shown to be a sound and complete axiomatisation with respect to the synchronous language model.

In the original SKA paper there is a presentation of \emph{synchronous Kleene algebra with tests} including a wrongful claim of completeness. An obvious next step would be to see if we can prove completeness of \SF with tests. We conjecture \SF with tests is indeed complete and that this is straightforward to prove via a reduction to \SF in a style similar to the completeness proof of KAT~\cite{kat}. Another generalisation is to add a unit to the semilattice, making it a bounded semilattice. This will probably lead to a type of delay operation~\cite{milner}.

Our original motivation to study SKA was to use it as an
axiomatisation of Reo, a modular language of connectors
combining synchronous data flow with an asynchronous one~\cite{BSAR2006}.
The semantics of Reo is based on an automata model very
similar to that of SKAT, in which transitions are labelled
by sets of ports (representing a synchronous data flow) and
constraints (the tests of SKAT). Interestingly, automata are
combined using an operation analogous to the synchronous
product of SKAT expressions. We aim to study the application of SKA or SKAT to Reo in future work.

\paragraph*{Acknowledgements}
The first author is grateful for discussions with Hans-Dieter Hiep and Benjamin Lion.

\bibliographystyle{plainurl}
\bibliography{ref}

\newpage
\appendix
\section{Appendix}
\synlanska*
\begin{proof}
The carrier $\synchronouslanguages$ is obviously closed under the operations of synchronous Kleene algebra. We need only argue that each of the SKA axioms is valid on synchronous languages.

The proof for the Kleene algebra axioms follows from the observation that synchronous languages over the alphabet $\bact$ are simply languages over the alphabet $\mathcal{P}_n(\bact)$. Thus we know that the Kleene algebra axioms are satisfied, as languages over alphabet $\mathcal{P}_n(\bact)$ with $1=\{\varepsilon\}$ and $0=\emptyset$ form a Kleene algebra.

For the semilattice axioms, note that $S$ is isomorphic to $\mathcal{P}_n(\Sigma)$ (by sending a singleton set in $S$ to its sole element), and that the latter forms a semilattice when equipped with $\cup$. Since the isomorphism between $S$ and $\mathcal{P}_n(\Sigma)$ respects these operators, it follows that $(S, \times)$ is also a semilattice.

The first SKA axiom that we check is commutativity. We prove that $\times$ on synchronous strings is commutative by induction on the paired length of the strings. Consider synchronous strings $u$ and $v$. For the base, where $u$ and $v$ equal $\varepsilon$, the result is immediate. In the induction step, we take $u=xu'$ with $x\in\mathcal{P}_n(\bact)$. If $v=\varepsilon$ we are done immediately. Now for the case $v=yv'$ with $y\in\mathcal{P}_n(\bact)$.
We have $u\times v= (xu')\times (yv')=(x\cup y)\cdot (u'\times v')$. From the induction hypothesis we know that $u'\times v' = v'\times u'$. Combining this with commutativity of union we have $u\times v = (x\cup y)\cdot (v'\times u') = v\times u$.
Take synchronous languages $K$ and $L$. Now consider $w\in K\times L$. This means that $w=u\times v$ for $u\in K$ and $v\in L$. From commutativity of synchronous strings we know that $w=u\times v= v\times u$. And thus we have $w\in L\times K$. The other inclusion is analogous.

It is obvious that the axioms $K\times \emptyset=\emptyset$ and $K\times \{\varepsilon\}=K$ are satisfied.

For associativity we again first argue that $\times$ on synchronous strings is associative. Take synchronous strings $u,v$ and $w$. We will show by induction on the paired length of $u,v$ and $w$ that $u\times (v\times w)= (u\times v)\times w$. If $u,v,w=\varepsilon$ the result is immediate. Now consider $u=xu'$ for $x\in\mathcal{P}_n(\bact)$. If $v$ or $w$ equals $\varepsilon$ the result is again immediate. Hence we consider the case where $v=yv'$ and $w=zw'$ for $y,z\in\mathcal{P}_n(\bact)$.
From the induction hypothesis we know that $u'\times (v'\times w')= (u'\times v')\times w'$.
We can therefore derive
\begin{align*}
u\times (v\times w)&= (xu') \times (yv'\times zw') = (xu') \times ((y\cup z)\cdot (v'\times w')) \\
&= (x\cup (y\cup z))\cdot (u'\times (v'\times w')) = (x\cup (y\cup z))\cdot ((u'\times v')\times w')
\end{align*}
From associativity of union, we then know that $(x\cup (y\cup z))\cdot ((u'\times v')\times w')=(u\times v)\times w$.
Now consider $t\in K\times(L\times J)$ for $K,L$ and $J$ synchronous languages. Thus $t=u\times (v\times w)$ for $u\in K$, $v\in L$ and $w\in J$. From associativity of synchronous strings we know that $t=u\times (v\times w)= (u\times v)\times w$, and thus we have $t\in (K\times L)\times J$. The other inclusion is analogous.

For distributivity consider $w\in K\times (L+J)$ for $K,L,J$ synchronous languages. This means that $w=u\times v$ for $u\in K$ and $v\in L+J$. Thus we know $v\in L$ or $v\in J$. We immediately get that $u\times v\in K\times L$ or $u\times v\in K\times J$ and therefore that $w\in K\times L + K\times J$. The other direction is analogous.

For the synchrony axiom we take synchronous languages $K,L$ and $A,B\in S$. Suppose $A=\{x\}$ and $B=\{y\}$ for $x,y\in\mathcal{P}_n(\bact)$. Take $w\in (A\cdot K)\times (B\cdot L)$. This means that $w=u\times v$ for $u\in A\cdot K$ and $v\in B\cdot L$.
Thus we know that $u=xu'$ with $u'\in K$ and $v=yv'$ with $v'\in L$. From this we conclude $w=u\times v=(xu')\times (yv')=(x\cup y)\cdot (u'\times v')$. As $u'\in K$ and $v'\in L$ and $x\cup y = x\times y$ with $x\in A$ and $y\in B$, we have that $w\in (A\times B)\cdot (K\times L)$.
For the other direction, consider $w\in (A\times B)\cdot (K\times L)$. This entails $w=t\cdot v$ for $t\in A\times B$ and $v\in K\times L$. As $A\times B=\{x\cup y\}$ we have $t=x\cup y$. And $v=u\times s$ for $u\in K$ and $s\in L$.
Thus $t\cdot v=(x\cup y)\cdot (u\times s)=(xu)\times (ys)$ for $u\in K$, $s\in L$, $x\in A$ and $y\in B$. Hence $w\in (A\cdot K)\times (B\cdot L)$.\qed
\end{proof}

\soundness*
\begin{proof}
This is proved by induction on the construction of $\equivska$. In the base case we need to check all the axioms generating $\equivska$, which we have already done for \autoref{lemma:synlan_ska}. For the inductive step, we need to check the closure rules for congruence preserve soundness. This is all immediate from the definition of the semantics of SKA and the induction hypothesis. For instance, if $e=e_0+e_1$, $f=f_0+f_1$, $e_0\equivska f_0$ and $e_1\equivska f_1$, then $\semska{e}=\semska{e_0}+\semska{e_1}=\semska{f_0}+\semska{f_1}=\semska{f}$, where use that $\semska{e_0}=\semska{f_0}$ and $\semska{e_1}=\semska{f_1}$ as a consequence of the induction hypothesis.
\end{proof}

\astar*
\begin{proof}
For the first inclusion, take $w\in\semska{a^* \times a^*} = \semska{a^*} \times \semska{a^*}$. Thus we have $w=u\times v$ for $u,v\in\semska{a^*}$. Hence $u=x_1\cdots x_n$ for $x_i\in\semska{a}$ and $v=y_1\cdots y_m$ for $y_i\in\semska{a}$.
As $\semska{a}=\{\semsem{a}\}$ with $\semsem{a}\in\mathcal{P}_n(\bact)$, we know that $x_i=\semsem{a}$ and $y_i=\semsem{a}$.
Assume that $n\leq m$ without loss of generality. We then know that $v=u\cdot \semsem{a}^{m-n}$, where synchronous string $e^n$ indicates $n$ copies of string $e$ concatenated. Unrolling the definition of $\times$ on words, we find $u\times v=u\times (u\cdot \semsem{a}^k)=(u\times u)\cdot \semsem{a}^k=u\cdot\semsem{a}^k=v$, and hence $w = u \times v = v \in \semska{a^*}$.
For the other inclusion, take $w\in\semska{a^*}$. As $\varepsilon\in\semska{a^*}$ and $w\times \varepsilon=w$, we immediately have $w\in\semska{a^*}\times\semska{a^*}$.\qed
\end{proof}

\begin{lemma}\label{infiniteness}
  For $K,L\in\countermodel$, $K$ a non-empty finite language and $L$ an infinite language, $K\times L$ is an infinite language.
\end{lemma}
\begin{proof}
Suppose that $K\times L$ is a finite language. Hence we have an upper bound on the length of words in $K\times L$. Since the length of the synchronous product of two words is obviously the maximum of the length of the operands, this means we also have an upper bound on the length of words in $L$, and as we have finite words over a finite alphabet in $L$ this means that $L$ is finite. Hence we get a contradiction, thus $K\times L$ is infinite.
\end{proof}

\countermodelis*
\begin{proof}
In the main text we treated one of the least fixpoint axioms and the synchrony axiom, and here we will treat all the remaining cases. For the sake of brevity, for each axiom we omit the cases where we can appeal to the proof for (synchronous) regular languages.

The proof that $(S, \times)$ is a semilattice is the same as in \autoref{lemma:synlan_ska}.
Next, we take a look at the Kleene algebra axioms.
If $K\in \countermodel$, then $K+\emptyset= \emptyset$ holds by definition of union of sets. If $K=\dagger$, we get $\dagger+\emptyset=\dagger$, and the axiom also holds.

For $K\in \countermodel\cup\{\dagger\}$, the axiom $K+K=K$ also easily holds by definition of the plus operator. Same for $K\cdot \{\varepsilon\}= K = K \cdot \{\varepsilon\}$ and $K\cdot \emptyset = \emptyset = \emptyset \cdot K$ by definition of the operator for sequential composition.

It is easy to see the axioms $1+e\cdot e^*\equivska e^*$ and $1+e^*\cdot e\equivska e^*$ hold for $K\in \countermodel$. In case $K=\dagger$, for $1+e\cdot e^*\equivska e^*$ we have
\[
1+\dagger\cdot\dagger^*= 1+\dagger\cdot\dagger=1+\dagger=\dagger=\dagger^*
\]
and a similar derivation for $1+e^*\cdot e\equivska e^*$.

For the commutativity of $+$ we take $K,L \in \countermodel\cup\{\dagger\}$.
If $K=\dagger$ or $L=\dagger$, we have $K+L=\dagger=L+K$.

For associativity of the plus operator we take $K,L,J\in \countermodel\cup\{\dagger\}$. If any of $K$, $L$ or $J$ is $\dagger$, it is easy to see the axiom holds.

For associativity of the sequential composition operator, consider $K,L,J\in \countermodel\cup\{\dagger\}$. We first can observe that if one of $K$, $L$ or $J$ is empty, then the equality holds trivially. Otherwise, if one of $K$, $L$ and $J$ is $\dagger$, then $(K\cdot L)\cdot J = \dagger = K\cdot (L\cdot J)$.

Next, we verify distributivity of concatenation over $+$. We will show a detailed proof for left-distributivity only; right-distributivity can be proved similarly. Let $K,L,J\in \countermodel\cup\{\dagger\}$. If one of $K$, $L$, or $J$ is empty, then the claim holds immediately (the derivation is slightly different for $K$ versus $L$ or $J$). Otherwise, if one of $K$, $L$ or $J$ is $\dagger$, then $K\cdot (L + J) = \dagger = K\cdot L + K\cdot J$.

For the remaining least fixpoint axiom, let $K,L,J\in \countermodel\cup\{\dagger\}$. Assume that $K+L\cdot J\leq L$. We need to prove that $K\cdot J^*\leq L$. If $L=\dagger$, then the claim holds immediately. If $L\in\countermodel$ and $J = \dagger$, then $L$ must be empty, hence $K$ is empty, and the claim holds. If $L,J\in\countermodel$, then also $K\in\countermodel$ and the proof goes through as it does for KA\@.

We now get to the axioms for the $\times$-operator. The commutativity axiom is obvious from the commutative definition of $\times$ (as we already know that $\times$ is commutative on synchronous strings). The axiom $K\times \emptyset=\emptyset$ is also satisfied by definition. The same holds for the axiom $K\times \{\varepsilon\}= K$ as $\{\varepsilon\}$ is finite.

For associativity of the synchronous product, consider $K,L,J\in \countermodel\cup\{\dagger\}$. If one of $K$, $L$ or $J$ is empty, then both sides of the equation evaluate to $\emptyset$. Otherwise, if one of $K$, $J$, or $L$ is $\dagger$, then both sides of the equation evaluate to~$\dagger$. If $K$, $J$ and $L$ are all languages, and at most one of them is finite, then either $K \times L = \dagger$, in which case the left-hand side evaluates to $\dagger$, or $K \times L$ is infinite (by \autoref{infiniteness}) and $J = \dagger$, in which case the right-hand side evaluates to $\dagger$ again. The right-hand side can be shown to evaluate to $\dagger$ by a similar argument.
In the remaining cases (at least two out of $K$, $J$ and $L$ are finite languages and none of them is $\dagger$ or $\emptyset$), the proof of associativity for the language model applies.

For distributivity of synchronous product over $+$, let $K,L,J\in \countermodel\cup\{\dagger\}$. If one of $K$, $L$ or $J$ is $\emptyset$, then the proof is straightforward. Otherwise, if one of $K$, $L$ or $J$ is $\dagger$, then both sides evaluate to $\dagger$. If $K$ and $L + J$ are infinite, then the outcome is again $\dagger$ on both sides (note that $L + J$ being infinite implies that either $L$ or $J$ is infinite).
In the remaining cases, $K$, $L$ and $J$ are languages and either $K$ or $L + J$ (hence $L$ and $J$) is finite. In either case the proof for synchronous regular languages goes through.
\qed
\end{proof}

\soundnessska*
\begin{proof}
We need to verify each of the axioms of \SF. The proof for the  axioms of \F is immediate via the observation that synchronous languages over the alphabet $\bact$ are simply languages over the alphabet $\mathcal{P}_n(\bact)$. Thus we know that the \F-axioms are satisfied, as languages over alphabet $\mathcal{P}_n(\bact)$ with $1=\{\varepsilon\}$ and $0=\emptyset$ form an \F-algebra. The additional axioms are the same as the ones that were added to KA for SKA, and we know they are sound from~\autoref{lemma:soundness}. 
\qed
\end{proof}

\reachisreach*
\begin{proof}
We prove the first statement by induction on the structure of $e$.
In the base, if we have $e \in \{0, 1\}$, the claim holds vacuously. If we have $a\in\bact$, then $\reach(a)=\{1,a\}$ and $\delta(a,A)=\{1 : A=\{a\}\}$, so the claim follows.
For the inductive step, there are five cases to consider.
\begin{itemize}
    \item If $e=H(e_0)$, then immediately $\delta(H(e_0),A)=\emptyset$ so the claim holds vacuously.
    \item
    If $e = e_0 + e_1$, then by induction we have $\delta(e_0, A) \subseteq \reach(e_0)$ and $\delta(e_1, A) \subseteq \reach(e_1)$.
    Hence, we find that $\delta(e, A) = \delta(e_0, A) \cup \delta(e_1, A) \subseteq \reach(e_0) \cup \reach(e_1) = \reach(e)$.
    \item
    If $e = e_0 \cdot e_1$, then by induction we have $\delta(e_0, A) \subseteq \reach(e_0)$ and $\delta(e_1, A) \subseteq \reach(e_1)$.
    Hence, we can calculate that
    \begin{align*}
    \delta(e, A)
        &= \{ e_0' \cdot e_1 : e_0' \in \delta(e_0, A) \} \cup \Delta(e_1,e_0,A) \\
        &\subseteq \{ e_0' \cdot e_1 : e_0' \in \reach(e_0) \} \cup \reach(e_1)
         = \reach(e)
    \end{align*}

    \item If $e = e_0 \times e_1$, then by induction we have $\delta(e_0, A) \subseteq \reach(e_0)$ and $\delta(e_1, A) \subseteq \reach(e_1)$ for all $A\in\mathcal{P}_n(\bact)$.
    Hence, we can calculate that
    \begin{align*}
    \delta(e, A)
        &= \{ e_0' \times e_1' : e_0' \in \delta(e_0, B_1), e_1'\in\delta(e_1,B_2),B_1\cup B_2=A \} \\
        & \quad \cup \Delta(e_0, e_1, A) \cup \Delta(e_1, e_0, A) \\
        &\subseteq \{ e_0' \times e_1' : e_0' \in \reach(e_0), e_1'\in\reach(e_1) \} \cup \reach(e_0)\cup \reach(e_1) = \reach(e)
    \end{align*}

    \item
    If $e = e_0^*$, then by induction we have $\delta(e_0, A) \subseteq \reach(e_0)$.
    Hence, we find that
    \[
        \delta(e, A) = \{ e_0' \cdot e_0^* : e_0' \in \delta(e_0, A) \} \subseteq \{ e_0' \cdot e_0^* : e_0' \in \reach(e_0) \}\subseteq\reach(e)
    \]
\end{itemize}

For the second statement, we prove that if $e' \in \reach(e)$, then $\reach(e') \subseteq \reach(e)$. The result of the first part tells us that $\delta(e',A)\subseteq\reach(e')$, which together with $\reach(e') \subseteq \reach(e)$ proves the claim.
We proceed by induction on $e$.
In the base, there are two cases to consider.
First, if $e =0$, then the claim holds vacuously. If $e=1$, then the only $e'\in\reach(e)$ is $e'=1$, so the claim holds. If $e=a$ for $a\in\bact$, we have $\reach(e)=\{1,a\}$. It trivially holds that $\reach(e')\subseteq \reach(e)$ for $e'\in\reach(e)$.

For the inductive step, there are four cases to consider.
\begin{itemize}
    \item
    If $e = H(e_0)$, then $\reach(e) = \{ 1 \}$, and the proof is as in the case where $e = 1$.

    \item
    If $e = e_0 + e_1$, assume w.l.o.g.\ that $e' \in \reach(e_0)$.
    By induction, we derive that
    \[
        \reach(e') \subseteq \reach(e_0) \subseteq \reach(e)
    \]

    \item
    If $e = e_0 \cdot e_1$ then there are two cases to consider.
    \begin{itemize}
        \item
        If $e' = e_0' \cdot e_1$ where $e_0' \in \reach(e_0)$, then we calculate
        \begin{align*}
            \reach(e')
                &= \{ e_0'' \cdot e_1 : e_0'' \in \reach(e_0') \} \cup \reach(e_1)\\
                &\subseteq \{ e_0'' \cdot e_1 : e_0'' \in \reach(e_0) \} \cup \reach(e_1)
                = \reach(e)
        \end{align*}
        \item
        If $e' \in \reach(e_1)$, then by induction we have $\reach(e') \subseteq \reach(e_1) \subseteq \reach(e)$.
    \end{itemize}

    \item
    If $e = e_0 \times e_1$ then there are three cases to consider.
    \begin{itemize}
        \item
        The first case is $e' = e_0' \times e_1'$ where $e_0' \in \reach(e_0)$ and $e_1'\in\reach(e_1)$, we get $\reach(e_0')\subseteq\reach(e_0)$ and $\reach(e_1')\subseteq\reach(e_1)$ by induction.
We calculate
        \begin{align*}
            \reach(e')
                &= \{ e_0'' \times e_1'' : e_0'' \in \reach(e_0'), e_1'' \in \reach(e_1') \} \cup\reach(e_0')\cup\reach(e_1') \\
                & \subseteq \{ e_0'' \cdot e_1'' : e_0'' \in \reach(e_0), e_1'' \in \reach(e_1) \} \cup\reach(e_0)\cup\reach(e_1) \\
                & = \reach(e)
        \end{align*}
        \item
        For $e' \in \reach(e_0)$, then by induction we have $\reach(e') \subseteq \reach(e_0) \subseteq \reach(e)$.
        \item
        For $e' \in \reach(e_1)$, the argument is similar to the previous case.
    \end{itemize}

    \item
    If $e = e_0^*$, then either $e' = 1$ or $e' = e_0' \cdot e_0^*$ for some $e_0' \in \reach(e_0)$.
    In the former case, $\reach(e') = \{1\}\subseteq \reach(e)$.
    In the latter case, we find by induction that
    \begin{align*}
    \reach(e')
        &= \{ e_0'' \cdot e_0^* : e_0'' \in \reach(e_0') \} \cup \reach(e_0^*) \\
        &\subseteq \{ e_0'' \cdot e_0^* : e_0'' \in \reach(e_0) \} \cup \reach(e_0^*) \subseteq  \reach(e_0^*)
        \tag*{\qed}
    \end{align*}
\end{itemize}
\end{proof}

\nfisnf*
\begin{proof}
  As $e\in\nfcross$ we have that $e=\overline{e_0}$ for some $e_0\in\cross$. From \autoref{factsnormalform} we know that $\overline{e_0}\equivsf e_0$. So we get $e\equivsf e_0$. Again from \autoref{factsnormalform}
  we then know that $\overline{e}=\overline{e_0}=e$.\qed
\end{proof}

\begin{lemma}%
  \label{lemma:pi-on-strings}
  For $x,y\in{(\mathcal{P}_n(\bact))}^*$, we have ${(x\cdot y)}^{\Pi}=x^{\Pi}\cdot y^{\Pi}$.
\end{lemma}
\begin{proof}
We proceed by induction on the lenth of $xy$. In the base, we have $xy=\varepsilon$. Thus $x=\varepsilon$ and $y=\varepsilon$. We have $\varepsilon^{\Pi}=\varepsilon$ so the result follows immediately.
In the inductive step we consider $xy=aw$ for $a\in\mathcal{P}_n(\bact)$. We have to consider two cases. In the first case we have $x=ax'$.
The induction hypothesis gives us that ${(x'\cdot y)}^{\Pi}=x'^{\Pi}\cdot y^{\Pi}$.
We then have ${(x\cdot y)}^{\Pi}={(ax'\cdot y)}^{\Pi}=a^{\Pi}\cdot{(x'\cdot y)}^{\Pi}=a^{\Pi}\cdot x'^{\Pi}\cdot y^{\Pi}=x^{\Pi}\cdot y^{\Pi}$.
In the second case we have $x=\varepsilon$ and $y=aw$. We then conclude that ${(x\cdot y)}^{\Pi}=y^{\Pi}= x^{\Pi}\cdot y^{\Pi}$.\qed
\end{proof}

\pionlanguages*
\begin{proof}
\begin{enumerate}[(i)]
    \item
    First, suppose $w \in {(L\cup K)}^{\Pi}$. Thus we have $w=v^{\Pi}$ for $v\in L\cup K$. This gives us $v\in L$ or $v\in K$. We assume the former without loss of generality. Thus we know $w=v^{\Pi}\in L^{\Pi}$. Hence we know $w\in L^{\Pi}\cup K^{\Pi}$. The other direction can be proved analogously.
    \item
      First, suppose $w \in {(L\cdot K)}^{\Pi}$. Thus we have $w=v^{\Pi}$ for some $v\in L\cdot K$. This gives us $v=v_1\cdot v_2$ for some $v_1\in L$ and some $v_2\in K$.
      By definition of ${(-)}^{\Pi}$ we know that $v_1^{\Pi}\in L^{\Pi}$ and $v_2^{\Pi}\in K^{\Pi}$. Thus we have $v_1^{\Pi}\cdot v_2^{\Pi}\in L^{\Pi}\cdot K^{\Pi}$.
      From \autoref{lemma:pi-on-strings} we know that $w=v^{\Pi}={(v_1\cdot v_2)}^{\Pi}=v_1^{\Pi}\cdot v_2^{\Pi}$, which gives us the desired result of $w\in L^{\Pi}\cdot K^{\Pi}$. The other direction can be proved analogously.
    \item
    Take $w\in {(L^*)}^{\Pi}$. Thus we have $w=v^{\Pi}$ for some $v\in L^*$. By definition of the star of a synchronous language we know that $v=u_1\cdots u_n$ for $u_i\in L$. As $u_i\in L$, we have $u_i^{\Pi}\in L^{\Pi}$ and $u_1^{\Pi}\cdots u_n ^{\Pi}\in {(L^{\Pi})}^*$.
    By \autoref{lemma:pi-on-strings}, we know that $w=v^{\Pi}={(u_1\cdots u_n)}^{\Pi}=u_1^{\Pi}\cdots u_n^{\Pi}$. Thus we have $w\in {(L^{\Pi})}^*$, which is the desired result. The other direction can be proved analogously.
    \qed
\end{enumerate}
\end{proof}

\fundamentaltheorem*
\begin{proof}
Here we treat the inductive cases not displayed in the main proof, where we treated only the synchronous case.
\begin{itemize}
    \item
    If $e=H(e_0)$, derive:
    \begin{align*}
    H(e_0)
        &\equivsf H(o(e_0)) + \sum_{e'\in\delta(e_0,A)} H(A^\Pi) \cdot H(e') \tag{IH, compatibility of $H$} \\
        &\equivsf H(o(e_0)) \tag{$H(A^\Pi) = 0$} \\ 
        &\equivsf o(H(e_0)) \tag{$o(H(e_0)) \in 2$} \\ 
        &\equivsf o(H(e_0)) + \sum_{e' \in \delta(H(e_0), A)} A^\Pi \cdot e'\tag{Def. $\delta$} \\
    \end{align*}

    \item
    If $e=e_0+e_1$, derive:
    \begin{align*}
    e_0+e_1
        &\equivsf \continuation(e_0)+\sum_{e'\in\delta(e_0,A)} A^{\Pi}\cdot e' + \continuation(e_1)+\sum_{e'\in\delta(e_1,A)} A^{\Pi}\cdot e'
            \tag{IH} \\
        &\equivsf \continuation(e_0+e_1)+\sum_{e'\in\delta(e_0,A)\cup\delta(e_1,A)} A^{\Pi}\cdot e'
            \tag{Def. $o$, merge sums} \\
        &\equivsf \continuation(e_0+e_1)+\sum_{e'\in\delta(e_0+e_1,A)} A^{\Pi}\cdot e'
            \tag{Def. $\delta$ }
    \end{align*}

    \item
    If $e=e_0\cdot e_1$, derive:
    \begin{align*}
    e_0\cdot e_1
        &\equivsf \big(\continuation(e_0)+\sum_{e'\in\delta(e_0,A)} A^{\Pi}\cdot e'\big) \cdot e_1
            \tag{IH} \\
        &\equivsf \continuation(e_0)\cdot e_1 +\sum_{e'\in\delta(e_0,A)} (A^{\Pi}\cdot e'\cdot e_1)
        \tag{Distributivity}   \\
        & \equivsf \continuation(e_0)\cdot \big(\continuation(e_1)+\sum_{e'\in\delta(e_1,A)} A^{\Pi}\cdot e'\big) +\sum_{e'\in\delta(e_0,A)} (A^{\Pi}\cdot e'\cdot e_1)
        \tag{IH} \\
        & \equivsf \continuation(e_0\cdot e_1) + \continuation(e_0)\cdot\sum_{e'\in\delta(e_1,A)} A^{\Pi}\cdot e' +\sum_{e'\in\delta(e_0,A)} (A^{\Pi}\cdot e'\cdot e_1)
        \tag{Def. $\continuation$, distributivity} \\
        & \equivsf \continuation(e_0\cdot e_1) + \sum_{e'\in\Delta(e_1,e_0,A)} A^{\Pi}\cdot e' +\sum_{e'\in\{e_0'\cdot e_1:e_0'\in\delta(e_0,A)\}} A^{\Pi}\cdot e' \\
        &\equivsf \continuation(e_0\cdot e_1)+\sum_{e'\in\delta(e_0\cdot e_1,A)} A^{\Pi}\cdot e'
            \tag{Def. $\delta$ }
    \end{align*}

    \item
    If $e=e_0^*$, we derive:
    \begin{align*}
      e_0^*
      &\equivsf  {\Big(\continuation(e_0)+\sum_{e'\in\delta(e_0,A)}A^{\Pi}\cdot e'\Big)}^* \tag{Induction hypothesis} \\
      &\equivsf {\Big(\hspace{-0.6cm}\sum_{\quad e'\in\delta(e_0,A)}\hspace{-0.6cm}A^{\Pi}\cdot e'\Big)}^* \tag{$\continuation(e_0) \in 2$ and loop tightening}\\ 
      &\equivsf 1+ \Big(\sum_{e'\in\delta(e_0,A)}A^{\Pi}\cdot e'\Big)\cdot {\Big(\sum_{e'\in\delta(e_0,A)}A^{\Pi}\cdot e'\Big)}^* \tag{star axiom of \SF} \\
      &\equivsf 1+ \Big(\sum_{e'\in\delta(e_0,A)}A^{\Pi}\cdot e'\Big)\cdot e_0^* \tag{first two steps} \\
      &\equivsf 1+ \sum_{e'\in\delta(e_0,A)}(A^{\Pi}\cdot e'\cdot e_0^*) \tag{Distributivity}\\
        &\equivsf \continuation(e_0^*)+ \sum_{e'\in\delta(e_0^*,A)}A^{\Pi}\cdot e' \tag*{(Def. $\continuation$, def. $\delta$) \qed}
    \end{align*}
\end{itemize}
\end{proof}

\semantictranslation*
\begin{proof}
In the main text we have treated the base cases. The inductive cases work as follows. There are three cases to consider. If $e=e_0+e_1$, then ${(\semska{e})}^{\Pi}={(\semska{e_0}\cup\semska{e_1})}^{\Pi}={(\semska{e_0})}^{\Pi}\cup {(\semska{e_1})}^{\Pi}$ (\autoref{lemma:pi-on-languages}).
From the induction hypothesis we obtain ${(\semska{e_0})}^{\Pi}=\semka{e_0}$ and ${(\semska{e_1})}^{\Pi}=\semka{e_1}$.
Combining these results we get ${(\semska{e})}^{\Pi}=\semka{e_0}\cup\semka{e_1}=\semka{e_0}+\semka{e_1}=\semka{e_0+e_1}$, so the claim follows.
Secondly, if $e=e_0\cdot e_1$, then ${(\semska{e})}^{\Pi}={(\semska{e_0}\cdot\semska{e_1})}^{\Pi}={(\semska{e_0})}^{\Pi}\cdot{(\semska{e_1})}^{\Pi}$ (\autoref{lemma:pi-on-languages}).
From the induction hypothesis we obtain ${(\semska{e_0})}^{\Pi}=\semka{e_0}$ and ${(\semska{e_1})}^{\Pi}=\semka{e_1}$.
We can then conclude that ${(\semska{e})}^{\Pi}=\semka{e_0}\cdot\semka{e_1}=\semka{e_0\cdot e_1}$.
Lastly, if $e=e_0^*$, we get ${(\semska{e_0^*})}^{\Pi}={({(\semska{e_0})}^*)}^{\Pi}={({(\semska{e_0})}^{\Pi})}^*$ (\autoref{lemma:pi-on-languages}).
From the induction hypothesis we obtain ${(\semska{e_0})}^{\Pi}=\semka{e_0}$.
Thus we have ${(\semska{e})}^{\Pi}=\semka{e_0}^* = \semka{e_0^*}$ and the claim follows.\qed
\end{proof}

\leastsolution*
\begin{proof}
We will construct $x$ by induction on the size of $Q$. In the base, let $Q=\emptyset$. In this case the unique $Q$-vector is a solution.
In the inductive step, take $k\in Q$ and let $Q'=Q\setminus\{k\}$. Then construct the $Q'$-linear system $(M', p')$ as follows:
\begin{align*}
  M'(i,j)&=M(i,k)\cdot {M(k,k)}^*\cdot M(k,j) + M(i,j)\\
  p'(i)&= p(i)+M(i,k)\cdot {M(k,k)}^*\cdot p(k)
\end{align*}
As $Q'$ is a strictly smaller set than $Q$ and $M'$ is guarded, we can apply our induction hypothesis to $(M', p')$. So we know by induction that $(M', p')$ has a unique solution $x'$. Moreover, if $M'$ and $p'$ are in normal form, so is $x'$; note that if $M$ and $p$ are in normal form, then so are $M'$ and $p'$.

We use $x'$ to construct the $Q$-vector $x$:
\[x(i)=\twopartdef{x'(i)}{i\neq k}{{M(k,k)}^*\cdot \big(p(k)+\sum_{j\in Q'}M(k,j)\cdot x'(j)\big)}{i=k}
\]
The first thing to show now is that $x$ is indeed a solution of $(M, p)$. To this end, we need to show that $M\cdot x+p\equivsf x$. We have two cases. For $i\in Q'$ we derive:
\begin{align*}
x(i)&=x'(i)\tag{Def. $x$} \\
& \equivsf p'(i)+\sum_{j\in Q'} M'(i,j)\cdot x'(j) \tag{$x'$ solution of $(M', p')$} \\ 
& \equivsf p(i)+M(i,k)\cdot {M(k,k)}^*\cdot p(k) \\
& \quad +\sum_{j\in Q'} (M(i,k)\cdot {M(k,k)}^*\cdot M(k,j) + M(i,j))\cdot x'(j) \tag{Def. $(M', p')$} \\ 
& \equivsf p(i)+\sum_{j\in Q'}M(i,j)\cdot x'(j) \\
& \quad + M(i,k)\cdot {M(k,k)}^*\cdot \big(p(k)+ \sum_{j\in Q'}M(k,j)\cdot x'(j)\big) \tag{Distributivity} \\
& \equivsf p(i)+\sum_{j\in Q'}M(i,j)\cdot x(j) + M(i,k)\cdot x(k)  \tag{Def. $x$} \\
& \equivsf p(i)+\sum_{j\in Q}M(i,j)\cdot x(j)  \tag{Merge sum}
\end{align*}

For $i=k$, we derive:
\begin{align*}
x(k) &= {M(k,k)}^*\cdot \big(p(k)+\sum_{j\in Q'}M(k,j)\cdot x'(j)\big) \tag{Def. $x$} \\
& \equivsf (1+M(k,k)\cdot {M(k,k)}^*)\cdot \big(p(k)+\sum_{j\in Q'}M(k,j)\cdot x'(j)\big) \tag{star axiom} \\
& \equivsf p(k)+\sum_{j\in Q'}M(k,j)\cdot x'(j) \\
& \quad +M(k,k)\cdot {M(k,k)}^* \cdot \big(p(k)+\sum_{j\in Q'}M(k,j)\cdot x'(j)\big) \tag{Distributivity} \\
& \equivsf p(k)+\sum_{j\in Q'}M(k,j)\cdot x(j) + M(k,k)\cdot x(k) \tag{Def. $x$} \\
& \equivsf p(k)+\sum_{j\in Q}M(k,j)\cdot x(j) \tag{Merge sum}
\end{align*}
We now know that $x$ is a solution to $(M, p)$ because $M\cdot x+p\equivsf x$. Furthermore, if $M$ and $p$ are in normal form, then so is $x'$, and thus $x$ is in normal form by construction.

Next we claim that $x$ is unique. Let $y$ be any solution of $(M, p)$. We choose the $Q'$-vector $y'$ by taking $y'(i)=y(i)$. To see that $y'$ is a solution to $(M', p')$, we first claim that the following holds:
\begin{equation}\label{equation:yk}
y(k)\equivsf {M(k,k)}^* \cdot \Big( p(k) + \sum_{j\in Q'} M(k,j)\cdot y(j)\Big)
\end{equation}

To see that this is true, derive
\begin{align*}
y(k) &\equivsf p(k)+\sum_{j\in Q}M(k,j)\cdot y(j) \tag{$y$ solution of $(M, p)$} \\ 
& \equivsf p(k) + M(k,k)\cdot y(k) + \sum_{j\in Q'}M(k,j)\cdot y(j)\tag{Split sum} \\
& \equivsf {M(k,k)}^* \cdot \Big( p(k) + \sum_{j\in Q'} M(k,j)\cdot y(j)\Big) \tag{Unique fixpoint axiom}
\end{align*}
Note that we can apply the unique fixpoint axiom because we know that $M$ is guarded and thus that $H(M(k,k))=0$.

Now we can derive the following:
\begin{align*}
y'(i)&=y(i)\tag{Def. $y$} \\
& \equivsf p(i)+\sum_{j\in Q} M(i,j)\cdot y(j) \tag{$y$ solution of $(M, p)$} \\ 
& \equivsf p(i)+ M(i,k)\cdot y(k)+\sum_{j\in Q'} M(i,j)\cdot y(j) \tag{Split sum} \\
& \equivsf p(i)+\sum_{j\in Q'} M(i,j)\cdot y(j) \\
& \quad + M(i,k)\cdot {M(k,k)}^* \cdot \Big( p(k) + \sum_{j\in Q'} M(k,j)\cdot y(j)\Big) \tag{\autoref{equation:yk}} \\
& \equivsf p(i)+M(i,k)\cdot {M(k,k)}^*\cdot p(k) \\
& \quad + \sum_{j\in Q'}\big(M(i,k)\cdot {M(k,k)}^*\cdot M(k,j) + M(i,j)\big)\cdot y(j) \tag{Distributivity} \\
& \equivsf p'(i)+\sum_{j\in Q'}M'(i,j)\cdot y(j)  \tag{Def. $(M', p')$} 
\end{align*}
Thus $y'$ is a solution to $(M', p')$. As $x'$ is the unique solution to $(M', p')$, we know that $y'\equivsf x'$.

For $i\neq k$ we know that $x(i)=x'(i)\equivsf y'(i)=y(i)$. For $i=k$ we can derive:
\begin{align*}
y(k)&\equivsf {M(k,k)}^* \cdot \Big( p(k) + \sum_{j\in Q'} M(k,j)\cdot y(j)\Big)\tag{\autoref{equation:yk}} \\
& \equivsf {M(k,k)}^* \cdot \Big( p(k) + \sum_{j\in Q'} M(k,j)\cdot y'(j)\Big) \tag{Def. $y'$} \\
& \equivsf {M(k,k)}^* \cdot \Big( p(k) + \sum_{j\in Q'} M(k,j)\cdot x'(j)\Big) \tag{$x'\equivsf y'$}  \\ 
& \equivsf {M(k,k)}^* \cdot \Big( p(k) + \sum_{j\in Q'} M(k,j)\cdot x(j)\Big) \tag{Def. $x'$} \\
& \equivsf x(k) \tag{Def. $x$}
\end{align*}
Thus, $y\equivsf x$, thereby proving that $x$ is the unique solution to $(M, p)$.
\qed
\end{proof}
\end{document}